\documentclass[3p]{elsarticle}
\pdfoutput=1
\usepackage{graphics,graphicx}
\usepackage{amsmath}
\usepackage{amsthm}
\usepackage{amssymb}
\usepackage{bm}
\usepackage[boxed,algoruled,linesnumbered,titlenotnumbered]{algorithm2e}

\newtheorem{definition}{Definition}
\newtheorem*{notation}{Notations}
\newtheorem{lemma}{Lemma}
\newtheorem{theorem}{Theorem}

\newtheoremstyle{examplestyle}%
  {9pt}{9pt}%
  {\normalfont\small}%
  {0pt}%
  {\normalfont\scshape}{.}%
  { }%
  { }
\theoremstyle{examplestyle}

\newtheorem{paradigma}{Example}
\newenvironment{example}[1][]{\begin{paradigma}[#1]\pushQED{$\triangle$}}{\hfill\popQED \end{paradigma}}

\makeatletter
\renewenvironment{proof}[1][\small\proofname]{\par
  \pushQED{\qed}%
  \normalfont\small \topsep6\p@\@plus6\p@\relax
  \trivlist
  \item[\hskip\labelsep
        \itshape
    #1\@addpunct{.}]\ignorespaces
}{%
  \popQED\endtrivlist\@endpefalse
}
\makeatother

\def\bibpath{/Users/dumas/Bib/Bib}

\def\bC{\mathbb{C}}

\def\bR{\mathbb{R}}

\def\cL{\mathcal{L}}

\def\cV{\mathcal{V}}

\def\cX{\mathcal{X}}

\mathchardef\slSigma="0106

\newcommand{\base}{B}
\newcommand{\length}[1]  {\lvert #1\rvert}
\newcommand{\abs}[1]  {\lvert #1\rvert}
\newcommand{\norm}[1]   {\left \|#1\right \|}

\newcommand{\Bint}[1]{(#1)_{{\base}}}
\newcommand{\Bexp}[1]{(0.#1)_{{\base}}}

\newcommand{\rhojsr}{\rho_{*}}
\newcommand{\sawtooth}[1]{\{#1\}}
\newcommand{\Id}{\operatorname{I}}
\newcommand{\matI}{\operatorname{I}}

\begin{document}

\begin{frontmatter}
\date{}
\title
{Rational series\\ and\\ asymptotic expansion\\ for\\ linear homogeneous \\ divide-and-conquer recurrences}
\author{Philippe Dumas}

\address{%
{\rm {\sc SpecFun} Team}, {\rm {\sc Inria}-Saclay},
France
}

\ead{Philippe.Dumas@inria.fr}

\begin{abstract}
Among all sequences that satisfy a divide-and-conquer recurrence, the sequences that are rational with respect to a numeration system are certainly the most immediate and most essential. Nevertheless, until recently they have not been studied  from the asymptotic standpoint. We show how a mechanical process permits to compute their asymptotic expansion. It is based on  linear algebra, with Jordan normal form, joint spectral radius, and dilation equations. The method is compared with the analytic number theory approach, based on Dirichlet series and residues, and new ways to compute the Fourier series of the periodic functions involved in the expansion are developed. The article comes with an extended bibliography. 
\end{abstract}

\end{frontmatter}

The aim of this article is the asymptotic study of a special type of divide-and-conquer sequences, namely the sequences which are rational with respect to a numeration system or \emph{radix-rational sequences}. In most cases where the numeration system is the binary system, such sequences satisfies a recursion which essentially links the values~$u_{2n}$ and~$u_{2n+1}$ to the value of~$u_n$ by a linear homogeneous relationship. Section~\ref{dumas-rrseq-v4:sec:RadixRationalSequences} is devoted to their precise definition. In Sections~\ref{dumas-rrseq-v4:sec:BasicIdeas} to~\ref{dumas-rrseq-v4:sec:AsymptoticExpansion}, we provide the reader  with a method of easy use to determine the asymptotic behaviour of radix-rational sequences. The used tools all come from linear algebra and this is the distinctive feature of this article. We insist on ideas and proofs are only sketched. A formal and complete treatment of this part can be found in~\cite{Dumas13}. It results in  Theorem~\ref{dumas-rrseq-v4:thm:Theorem}, which asserts that a radix-rational sequence admits an asymptotic expansion with variable coefficients in the scale 
$N^\alpha \binom{\log N}{m}$, or equivalently 
$N^\alpha \log^m N$. The so-called \emph{dichopile algorithm}, briefly reminded in Example~\ref{dumas-rrseq-v4:ex:Dichopile-1} below, will be a thread in this paper and its asymptotic behaviour is given by 
\begin{equation*}
	f_N \mathop{=}_{N \to + \infty}
	\frac{1}{2} N \log_2 N + N \Phi(\log_2 N) + O(N^{\varepsilon}),
\end{equation*}
where~$\Phi(t)$ is a $1$-periodic function and $0 < \varepsilon < 1$. We present in Section~\ref{dumas-rrseq-v4:sec:ExamplesAndImprovements} some  improvements of the method and we summarize it into an algorithm. 
Next in  Section~\ref{dumas-rrseq-v4:sec:LinearvsAnalytic}, we briefly mention the previous works about the asymptotic behaviour of divide-and-conquer sequences. Specifically, we show how our algebraic approach can help the analytic number theory approach, popularized by Philippe Flajolet. To end, in Section~\ref{dumas-rrseq-v4:sec:FourierCoefficients} we elaborate on the computation of the Fourier coefficients of the involved periodic functions and we provide the reader with several methods. This leads us to consider a Mellin transform.

\section{Radix-rational sequences}\label{dumas-rrseq-v4:sec:RadixRationalSequences}

Radix-rational sequences  are a mere generalization of classical rational sequences, that is sequences which satisfy a linear homogeneous recurrence with constant coefficients. In the same manner they satisfy a linear homogeneous recurrence with constant coefficients but where the shift $n\mapsto n+1$ is replaced by the pair of scaling transformations $n\mapsto 2n$ and $n\mapsto 2n+1$, in case the radix is~$2$. 

These sequences come to light in various domains of knowledge and the first example which comes in mind~\citep{Trollope68,Delange75} is the binary sum-of-digits function, that is the number~$s_n$ of~$1$'s in the binary expansion
of an integer~$n$~\cite{DrGa98}. 
The sequence satisfies $s_0=0$, $s_{2n}=s_n$,
$s_{2n+1}=s_n+1$. It may seem to the reader that such a sequence
is of limited interest, but it appears in many problems like the study
of the maximum of a determinant of a $n\times n$ matrix with
entries~$\pm1$~\cite{ClLi65} or in a merging process occurring in
graph theory~\cite{McIlroy74}. This example have been greatly
generalized with the number of occurrences of some pattern in the
binary code~\cite{BoCoMo89}, or in the Gray code: Flajolet and Ramshaw~\cite{FlRa80}
study the average case of Batcher's odd-even merge by using the
sum-of-digits function of the Gray code.  Among the sequences directly
related to a numeration system the Thue-Morse sequence which writes
$u_n=(-1)^{s_n}$ is certainly the one which has caused the
greatest number of publications~\cite{AlSh99}. There exist variants
with some subsequences~\cite{Newman69,Coquet83} or with binary
patterns~\cite{BoCoMo89} other than the simple pattern~$1$: the
Rudin-Shapiro sequence associated to the pattern~$11$ satisfies $u_{2n} = u_{4n+1} = u_{n}$, $u_{4n+3} = - u_{2n+1}$, $u_0 = u_1 =1$, and 
 was initially
designed to minimize the $L^{\infty}$ norm of a sequence of
trigonometric polynomials with
coefficients~$\pm1$~\cite{Rudin59,Shapiro51}.  The study of the
complexity of algorithms is another source of radix-rational
sequences. The cost~$c_n$ of computing the $n$-th power of a matrix by
binary powering satisfies $c_0=0$, $c_{2n}=c_n+1$,
$c_{2n+1}=c_n+2$. The idea of binary powering has been re-employed
by Morain and Olivos~\cite{MoOl90} in the context of computations on an elliptic curve
where the subtraction has the same cost than addition. Supowit and Reingold~\cite{SuRe83}
have used the divide-and-conquer strategy to give heuristics for
the problem of Euclidean matching and this leads them to a
$4$-rational sequence in the worst case. The theory of numbers is
another domain which provides examples like the number of odd binomial
coefficients in row~$n$ of Pascal's triangle~\cite{Stolarsky77} or
the number of integers which are a sum of three squares in the first
$n$ integers~\cite{OsSh89}.

It is the merit of \citet{AlSh92,AlSh03} to have put all these scattered examples in a common framework, that leads to the idea of a linear representation with matrices. For the examples above the matrix dimension~$d$ is usually small, say from~$1$ to~$4$, but there exist examples where $d$ is larger, like in the work of Cassaigne~\cite{Cassaigne93} which uses $d=30$. In such cases a general method of study is necessary and we will elaborate one. But, it is first necessary to formalize the notion of a radix-rational sequence.

\subsection{{Forward direction}}

Let~$S$ be a rational formal series over a finite alphabet~$\cX$. For the sake of simplicity, we are assuming here and in all the sequel that the coefficients of the formal series are complex numbers, even if it is possible to consider in this first section the more general framework of a commutative field or even of a commutative semi-ring. 
If~$\cX$ is the set of figures for the numeration system with radix~$\base$, the rational series defines a sequence~$s_n$ in the following manner: for each non-negative integer~$n$, we consider the word $w = n_{\ell-1}\dotsb n_1n_0$ which is the radix~$\base$ expansion of~$n$, that is 
\begin{equation}\label{dumas-rrseq-v4:eq:RadixExpansionDefinition}
n = (n_{\ell-1}\dotsb n_1n_0)_\base = n_0 + n_1\base +\dotsb + n_{\ell-1} \base^{\ell-1},
\end{equation}
and the term~$s_n$ of the sequence is the value $(S,w)$ of the rational series over the word~$w$.

More concretely, if a linear representation of the rational series~$S$ is at our disposal, we can compute the value $(S,w)$ for a word $w = w_1\dotsb w_{\ell}$ by the formula 
\begin{equation}\label{dumas-rrseq-v4:eq:RationalSeriesComputation}
	(S,w) = L A_{w_1}A_{w_2}\dotsb A_{w_{\ell}}C = LA_wC.
\end{equation}
Here, $(L,\ (A_b)_{0 \leq b < \base},\  C)$ is a linear representation of~$S$. This means that $L$ is a row vector, matrices $A_b$ are square matrices, and $C$ is a column vector, all with a coherent size~$d$, which is the dimension of the representation. Formula~\eqref{dumas-rrseq-v4:eq:RationalSeriesComputation} translates immediately into
\begin{equation}\label{dumas-rrseq-v4:eq:RadixRationalSequenceComputation}
	s_n = L A_{n_{\ell-1}}\dotsb A_1A_0C
\end{equation}
and gives a way to compute the successive values of the sequence~$s$ associated with the rational series~$S$.

\medskip

The rational character of the formal series~$S$ has the following meaning~\cite{BeRe88,Sakarovitch09}: there is a vector space which contains the series~$S$; this vector space  is of finite dimension~$d$; and this vector space is left stable by the trimming operators~$T_r$, with $0\leq r < \base$, defined by (the figure~$r$ is viewed as a letter)
\begin{equation*}
	T_r S = \sum_{\begin{subarray}{c} w = v r\end{subarray}} (S,w) v.
\end{equation*}
The operator~$T_r$ extracts the part of the formal series associated with the words which end with~$r$ and trims the letter~$r$. These operators translate immediately, via the numeration system, into the multisection operators $T_r s_n = s_{\base n + r}$. The correspondence is more clear if we use the ordinary generating function of the sequence~$s_n$, which parallels the formal series~$S$, 
\begin{equation}\label{dumas-rrseq-v4:eq:TrimmingOperatorsDefinition}
	T_r \sum_{n=0}^{+\infty} s_n z^n = \sum_{n=0}^{+\infty}s_{\base n + r} z^k.
\end{equation}
With this viewpoint, the linear representation of the formal series~$S$ is interpreted as follows for the sequence~$s_n$: we have a finite-dimensional vector space~$\cV$ of sequences, equipped with a basis $(v^1,\,v^2,\,\ldots,\,v^d)$ (or more generally a generating system), which is left stable by the multisection operators and contains the sequence~$s$. The column vector~$C$ gives the coordinates of the sequence~$s$ with respect to the basis~$(v^i)_{1\leq i \leq d}$; the matrix~$A_b$ is the expression of the multisection operator~$T_b$; and the row vector~$L$ appears to be the linear form $u\mapsto u_0$, which evaluates the sequences for $n = 0$. As an example, with $\base = 2$, $n = 5 = (101)_2$, the column vectors $C$, $A_1C$, $A_0A_1C$, and $A_1A_0A_1C$ are the expressions of the sequences $s_n$, $s_{2n+1}$, $s_{4n+1}$, and $s_{8n+5}$ respectively. Taking the value for $n = 0$ of the last sequence provides the value~$s_5$ and this explains the formula $s_5 = L A_1A_0A_1C$. 

\medskip

We are led to the following definition, but with a slight change with respect to~\cite{AlSh92,AlSh03}. We prefer the attributive adjective \emph{ rational} to \emph{ regular}, because \emph{ rational} is understood by both computer scientists and mathematicians, while \emph{ regular} is understood only by the former but remains  hazy for the latter. 

\begin{definition}\label{dumas-rrseq-v4:def:RadixRationalSequence}
A (complex) sequence is said to be rational with respect to a numeration system, or more briefly radix-rational, or more precisely $\base$-rational if its orbit under the action of the $\base$-multisection operators remains in a finite dimensional vector space.
\end{definition}

\vskip 2ex

\subsection{{Reverse direction}}
Let us assume now that we have a sequence~$s = (s_n)$ which is $\base$-rational. We can build a linear representation of this sequence in the following way. If the sequence is the null sequence there is nothing to do, and we consider that the associated dimension~$d$ is~$0$. Otherwise we put~$v^1 = s$ as the first vector of the basis we want to build. Next we consider successively the sequences $(s_{\base n})$, $(s_{\base n + 1})$, $\ldots$, $(s_{\base n + \base -1})$, $(s_{\base^2 n})$, $(s_{\base^2 n +1})$, and so on. For each of these sequences, either it is independent of the previous ones and we add it to the basis, say as~$v^k$, or it can be expressed as a linear combination of the previous sequences that are in the basis. The process necessarily stops since all these sequences live in a finite dimensional space. When it stops, we have a free family~$(v^1,\,v^2,\,\ldots,\,v^d)$ which is a basis of a vector space~$\cV$. This vector space is left stable by the multisection operators and contains the sequence~$s$. Moreover we have a linear representation of the sequence. First the sequence~$s$ is the first vector of the basis, and this gives the vector column $C = E_1$, say. Next we have noted in passing the dependencies encountered in the process, between each of the sequence~$v^j$ and its images by the multisection operators~$T_0v^j$, $T_1v^j$, $\ldots$, $T_{\base -1 }v^j$, and this gives us the matrices~$A_0$, $A_1$, $\ldots$, $A_{\base -1}$. And last the values~$v^1_0$, $\ldots$, $v^d_0$ at~$0$ give the row vector~$L$. 

By this process, we have associated with the $\base$-rational sequence a rational formal series over the alphabet of the radix~$\base$ numeration system. Moreover, we see that the classical rational sequences, like the Fibonacci sequence, that is the sequences whose ordinary generating function is a rational function (for which $0$ is not a pole), appear to be $1$-rational because it is well known that such a sequence writes $s_n = L A^n C$ for some square matrix~$A$. Hence, the $\base$-rational sequences generalize the classical rational sequences.

\vskip 2ex

\subsection{{Sensitivity to the leftmost zeroes}}
The above construction is not entirely satisfactory, for if a rational series determines a radix rational sequence, the converse is not true. The reason is that the sequence does not use the words which begin with some zeroes. (The $\base$-ary word associated with the number~$0$ is the empty word.) A first idea to fill this gap is to extend the formal series by the value 0 for words that begin with zero. A more preferable method is to follow the process of the previous paragraph. It provides sequences $v^j$ which satisfy obviously $(T_0 v^j)_0 = v^j_0$, that is, in one formula, $LA_0 = L$. This property determines completely the formal series if its values for the $\base$-ary expansions of the integers are known.

\begin{definition}\label{dumas-rrseq-v4:def:ZeroInsensitiveLinearRepresentation}
	A linear representation $(L,\ (A_b)_{0 \leq b < \base},\  C)$ of a $\base$-rational sequence is said to be insensitive to the leftmost zeroes, or simply zero-insensitive, if it satisfies $LA_0 = L$.
\end{definition}
This definition implies two remarks. First, a radix-rational sequence admits always a linear representation which is  insensitive to the leftmost zeroes. Above, we have proved this property. Second, because all the minimal linear representations of a rational formal series are isomorphic, the same property is satisfied for the radix-rational sequence if we use only zero-insensitive linear representations. 

\begin{example}[Dichopile algorithm]\label{dumas-rrseq-v4:ex:Dichopile-1}
The \emph{dichopile algorithm} is an attempt to find a balance between space and time in the random generation of words from a regular language~\cite{Oudinet10,OudeGa12} (the recursion~\eqref{dumas-rrseq-v4:eq:DichopileRecursion} below is rather hidden in the last reference). To generate uniformly a word of length~$n$ the previously known methods use $O(n)$ in space and $O(n)$ in time, while the dichopile algorithm uses only $O(\log n)$ in space (a big saving) but $O(n\log n)$ in time (a small loss), as we will see. The time complexity is given by the sequence $f_n$ defined as follows,
\begin{equation}\label{dumas-rrseq-v4:eq:DichopileRecursion}
	\begin{split}
	f_{n} & = n + g_n,\\
	g_{n} & = f_{\lfloor n/2 \rfloor -1} + g_{\lceil n/2\rceil},\qquad \text{for $n \geq 2$,}
	\end{split}
\end{equation}
with $f_1 = 1$, $g_1 = 0$. We add $f_0 = 0$, $g_0 = 0$. Both sequences $f_n$ and $g_n$ are rational with respect to the radix $B = 2$. 

As it will soon appear, we are not interested by a linear representation for the sequence~$f_n$ but for the sequence of backward differences $u_n = \nabla f_n = f_{n} - f_{n-1}$ (with $f_{-1} = 0$), which is $2$-rational too according to Lemma~\ref{dumas-rrseq-v4:lemma:BackwardDifferences} below. A linear representation of dimension~$6$ is the following 
\begin{multline}\label{dumas-rrseq-v4:eq:DichopileLinearRepresentation}
	L = \left( \begin {array}{cccccc} 0&0&0&0&1&0\end {array} \right),\quad
	C = \left( \begin {array}{cccccc} 1&0&0&0&0&0\end {array} \right)^T,\\
	A_0 = 
	\left( \begin {array}{cccccc} 0&0&0&0&0&0\\ \noalign{\medskip}1&0&0&1&0&0\\ \noalign{\medskip}0&0&1&0&0&0\\ \noalign{\medskip}0&1&0&0&0&0\\ \noalign{\medskip}0&0&0&0&1&0\\ \noalign{\medskip}1&1&0&0&0&1\end {array} \right),\quad
	A_1 =
	\left( \begin {array}{cccccc} 0&0&1&0&0&0\\ \noalign{\medskip}0&1&0&0&0&0\\ \noalign{\medskip}1&0&0&1&0&0\\ \noalign{\medskip}0&0&0&0&0&0\\ \noalign{\medskip}1&0&0&0&1&1\\ \noalign{\medskip}0&1&0&0&0&0\end {array} \right).
\end{multline}
It uses the generating family $u_n = \nabla f_n = 1 + \nabla g_n$, $\nabla f_{n-1} = 1 + \nabla g_{n-1}$, $\nabla f_{n+1} = \Delta g_n$, $\nabla g_n$, $1$, and $1 - \delta_{0,n}$ (here $\delta_{x,y}$ is the usual Kronecker symbol). As a consequence it is not a reduced linear representation, because the first sequence is the sum of the fourth and the fifth. Nevertheless it is perfectly acceptable and, moreover, it is insensitive to the leftmost zeroes. 
\end{example}

In the rest of the article, we will use only zero-insensitive linear representations.

\subsection{Backward differences}
As it is pointed out in Example~\ref{dumas-rrseq-v4:ex:Dichopile-1}, we do not use a linear representation of the sequence~$s_n$ we want to estimate asymptotically, but a linear representation of the sequence $u_n = \nabla s_n = s_n - s_{n-1}$ of its backward differences. The next assertion shows that the latter is radix-rational if the former is.

\begin{lemma}\label{dumas-rrseq-v4:lemma:BackwardDifferences}
	If a sequence~$s_n$ is $\base$-rational, then the sequence of its backward differences is $\base$-rational too. Conversely, if a sequence~$u_n$ is $\base$-rational, then the sequence of its partial sums $s_n$ is $\base$-rational too.
\end{lemma}

\begin{proof}
	We could give a proof based on linear representations as in~\cite[Lemma~1]{Dumas13}, but this is useless because we do not will employ the change from one representation to the other. It is simpler to consider the ordinary generating functions
	\[
	u(x) = \sum_{n=0}^{+\infty} u_n x^n,\qquad s(x) = \sum_{n=0}^{+\infty} s_n x^n,
	\]
which are related by
\[
u(x) = (1-x) s(x),\qquad s(x) = \frac{1}{1-x} u(x).
\]
Allouche and Shallit have proved that the product of two $\base$-rational generating functions is $\base$-rational, hence the title \emph{The \emph{ring} of $k$-regular sequences} of their article~\cite{AlSh92}. Moreover $1-x$ is $\base$-rational as a polynomial, and $1/(1-x)$ is $\base$-rational as a rational function whose poles are all roots of unity~\cite[Th.~3.3]{AlSh92}. Accordingly, if~$s(x)$ is $\base$-rational, so is~$u(x)$, and conversely.
\end{proof}

\section{Basic ideas}\label{dumas-rrseq-v4:sec:BasicIdeas}
Our aim is to expose a calculation method for an asymptotic expansion of a radix rational sequence~$s_n$. It takes as input a linear representation of the sequence of backward differences $u_n = \nabla s_n = s_n - s_{n-1}$ (with $s_{-1} = 0$). Its output is an asymptotic expansion for the sequence~$s_n$, which writes as a finite linear combination of terms
\begin{equation}\label{dumas-rrseq-v4:eq:BasicTermExpansion}
	n^\alpha \binom{\log_{\base} n}{m} \Psi(\log_{\base} n).
\end{equation}
Here~$\alpha$ is a real number, $\base$ is the radix of the numeration system, $m$ is a nonnegative integer, and $\Psi(t)$ is an oscillatory function often $1$-periodic in practice.

In the rest of the article, we will constantly use the following notations.
\begin{notation}
\noindent
Throughout,  $L$, $(A_b)_{0 \leq b < \base}$, $C$ is a zero-insensitive linear representation for $u_n = \nabla s_n$, which is $d$-dimensional. The matrix~$\Id_d$ is the identity matrix of size~$d$. The matrix~$Q$ is the sum of the square matrices in the representation,
\begin{equation}\label{dumas-rrseq-v4:eq:DefinitionQ}
	Q = \sum_{0 \leq b < \base} A_b.
\end{equation}
The integer~$K$ is the integer part of the logarithm base~$\base$ of~$N$, and~$t$ is its fractional part,
\begin{equation}\label{dumas-rrseq-v4:eq:Notation-K,t}
	K = \lfloor \log_\base N\rfloor,\qquad t = \log_\base N - K = \{\log_\base N\}.
\end{equation}
At last $\Bexp{w}$ is the number, in $\left[0,1\right)$, whose $\base$-ary expansion reads $0.w$.
\end{notation}

The main idea is to link the partial sum of~$u_n$ up to~$N$, that is~$s_N$, with the sum of the rational series~$S$ behind~$u_n$ for some words of a given length.
\begin{lemma}\label{dumas-rrseq-v4:lemma:BasicEquation}
Both sums
\begin{equation}\label{dumas-rrseq-v4:def:Sums}
	s_N = \sum_{n=0}^N u_n,\qquad \text{and}\qquad S_K(x) = \sum_{\begin{subarray}{c}\length{w} = K\\ \Bexp{w} \leq x\end{subarray}} (S,w),
\end{equation} 
are related by
\begin{equation}\label{dumas-rrseq-v4:eq:BasicEquationSimplified}
	s_N =  S_{K+1}(\base^{t-1}).
\end{equation}
\end{lemma}

\begin{proof}
Formula~\eqref{dumas-rrseq-v4:eq:BasicEquationSimplified} comes from cutting the interval $[0,N]$ by the powers of~$\base$. The integers in the interval $\left[\base^k,\base^{k+1}\right)$ have a $\base$-ary expansion which is a length~$k+1$ word. The sum over all length~$k+1$ words expresses with the matrix~$Q$, as follows
\[
	\sum_{\length{w} = k+1} L A_w C = \sum_{\length{w} = k+1} L A_{w_1} \dotsb A_{w_{k+1}} C = 
	L (A_0 + \dotsb + A_{\base -1})^{k+1} C = L Q^{k+1} C.
\]
This leads us to the formula
\begin{equation}\label{dumas-rrseq-v4:eq:BasicEquation}
	s_N = L(\Id_d-A_0)\sum_{k=0}^K Q^k C + S_{K+1}(\base^{t-1}).
\end{equation}
More precisely, the $\base$-ary expansions of the integers in $\left[\base^k,\base^{k+1}\right)$ do not begin with a~$0$, hence we have to subtract $L A_0 Q^k C$ from $L Q^{k+1} C$. This explains the first term in~\eqref{dumas-rrseq-v4:eq:BasicEquation}. The second term $S_{K+1}(\base^{t-1})$ corresponds to the interval $\left[\base^K,N\right]$. With the notations~\eqref{dumas-rrseq-v4:eq:Notation-K,t}, the integer $N$ writes $N = \base ^{K+t} = \base^{K+1}\base^{t-1}$ and for a length $K+1$ word~$w$ the inequality $\Bint{w} \leq N$ is equivalent to $\Bexp{w} \leq \base^{t-1}$, hence the occurrence of the term $S_{K+1}(\base^{t-1})$. At this point, we take advantage of our assumption that we consider only zero-insensitive linear representations and Formula~\eqref{dumas-rrseq-v4:eq:BasicEquation} simplifies into~\eqref{dumas-rrseq-v4:eq:BasicEquationSimplified}.
\end{proof}

It will turn out that it is not sufficient to consider the scalar-valued sum~$S_K(x)$, and we reinforce the notations with a vector-valued sum.
\begin{notation}
The sum $S_K(x)$ writes $L\slSigma_K(x)$ where $\slSigma_K(x)$ is the vector-valued sum
\begin{equation}\label{dumas-rrseq-v4:eq:VectorValuedSumDefinition}
	\slSigma_K(x) = \sum_{\begin{subarray}{c}\length{w} = K\\ \Bexp{w} \leq x\end{subarray}} A_w C
\end{equation}
\end{notation}

\noindent
We are interested by the behaviour of $\slSigma_K(x)$ for~$K$ large. Our starting point  is the next formula.
\begin{lemma}\label{dumas-rrseq-v4:lemma:RecursionFormula}
	The sum~$\slSigma_K(x)$ satisfies the recursion 
\begin{equation}\label{dumas-rrseq-v4:eq:RecursionFormula}
	\slSigma_{K+1}(x) = \sum_{b_1 < x_1} A_{b_1} Q^K C + A_{x_1} \slSigma_K (\base x - x_1),
\end{equation}
valid for each $ x = \Bexp{x_1x_2\ldots}$ in $\left[0,1\right)$. 
\end{lemma}

\begin{proof}
The numbers from $[0,x]$ with a $\base$-mantissa of length $K+1$ are sorted according to the prefixes of their $\base$-ary expansions. This emphasizes the intervals $\left[0,\Bexp{x_1}\right)$, $\left[\Bexp{x_1},\Bexp{x_1x_2}\right)$, $\ldots$ $\left[\Bexp{x_1x_2\ldots x_K},\Bexp{x_1x_2\ldots x_{K+1}}\right]$, and leads to the  formula
\[
\slSigma_{K+1}(x) = \sum_{b_1 < x_1} A_{b_1} Q^K C + \sum_{b_2 < x_2} A_{x_1} A_{b_2} Q^{K-1} C + \dotsb + \sum_{b_{K+1} \leq x_{K+1}} A_{x_1} A_{x_2} \dotsb A_{b_{K+1}} C.
\]
The above recursive Formula~\eqref{dumas-rrseq-v4:eq:RecursionFormula} is a mere consequence.
\end{proof}

We will decompose the column vector~$C$ of the linear representation on a basis of generalized eigenvectors for the matrix~$Q$, because for such vectors Formula~\eqref{dumas-rrseq-v4:eq:RecursionFormula} will reveal a functional equation, known as a dilation equation. But this will be possible only for eigenvalues sufficiently large, and we will begin by what means `sufficiently large' in the context. 

\section{Joint spectral radius}\label{dumas-rrseq-v4:sec:JointSpectralRadius}

The computing of a rational series uses products of square matrices of an arbitrary length. We need to evaluate the size of such products. We consider all products~$A_w$ for words~$w$ of a given length~$T$ and their norms, for some subordinate norm. Then the joint spectral radius of the set of matrices $A_b$, $0 \leq b < \base$, is defined as~\cite{RoSt60,Blondel08}
\begin{equation}\label{dumas-rrseq-v4:def:JointSpectralRadius}
\rhojsr = \lim_{T\to +\infty} \max_{|w|=T} \norm{A_w}^{1/T}.
\end{equation}
It is independent of the used subordinate norm. Moreover the maximum in the right hand side term is a non increasing function of~$T$. More information can be found in~\cite{BlKaPrWi08} and the papers which come with it.

\medskip

We apply immediately this definition to the study of the error term in the asymptotic expansion.
\begin{lemma}\label{dumas-rrseq-v4:lemma:ErrorTerm}
	Let~$V$ be a generalized eigenvector associated with the eigenvalue $\rho \omega$, $\rho \geq 0$, $\abs{\omega} = 1$, which satisfies $\rho \leq \rhojsr$. Then for every ${r} > \rhojsr$, the sequence $\slSigma_K$ associated with~$V$ satisfies $\norm{\slSigma_K} = O({r}^K)$.
\end{lemma}

\begin{proof}[Sketch of the proof]
Let us consider simply the case of an eigenvector and assume $\norm{A_b} \leq r$ for $1 \leq b < \base$. With Formula~\eqref{dumas-rrseq-v4:eq:RecursionFormula} of Lemma~\ref{dumas-rrseq-v4:lemma:RecursionFormula}, we readily obtain
\[
\norm{\slSigma_{K+1}(x)} \leq \rho^K \sum_{b_1 < x_1} \norm{A_{b_1}} \norm{V} + \norm{A_{x_1}} \norm{\slSigma_K(\base x - x_1)}
\]
and, using the supremum norm,
\[
\norm{\slSigma_{K+1}}_\infty \leq \rho^K \base {r} \norm{V} + {r} \norm{\slSigma_K}_\infty.
\]
With $\rho < {r}$, this recursion leads to $\norm{\slSigma_K}_\infty = O({r}^K)$. More generally, we can find an integer~$T$ such that for all words~$w$ of length~$T$ we have $\norm{A_w} \leq r$ and we can deal with the subsequences $\norm{\slSigma_{KT+r}(x)}$, $0 \leq r < T$, as previously. 
\end{proof}

\begin{example}[Dichopile algorithm, continued]\label{dumas-rrseq-v4:ex:Dichopile-2}
Here, we are using the maximum absolute column sum~$\norm{\ }_1$. Because~$A_0$ and~$A_1$ have nonnegative coefficients, computing  $\norm{A_w}_1$ amounts to compute the row vector $U A_w$ and take the largest component, where~$U$ is the row vector with all components equal to~$1$. It turns out that all these products $U A_w$ write 
\[\displaystyle W = \left( \begin {array}{cccccc} a&b&c&d&1&1\end {array} \right) \]
because this is true for the empty word, and this writing is preserved in the product by~$A_0$ and~$A_1$,
\[\displaystyle  WA_0 = \left( \begin {array}{cccccc} 1+b&1+d&c&b&1&1\end {array} \right) ,\, WA_1 = \left( \begin {array}{cccccc} 1+c&1+b&a&c&1&1\end {array} \right) \]
The formul{\ae} above show that all the components for a word of length~$T$ are bounded by $T+1$. Moreover if $UA_1^T$ (here~$T$ is an exponent) writes $[a,b,c,d,1,1]$ where~$b$ is larger than all the other components, then $UA_1^{T+1}$ writes $[1+c,1+b,a,c,1,1]$ and $1+b$ is the largest component, so that $\norm{UA_1^T} = T+1$. We  obtain
\[
\max_{\length{w} = T}\norm{A_w}_1 = T+1.
\]
As a consequence we conclude $\rhojsr = 1$%
, and the error term will be~$O(r^K)$ for every $r > 1$.
\end{example}

\section{Dilation equations}\label{dumas-rrseq-v4:sec:DilationEquations}

Assume for a while that we take for~$C$ an eigenvector~$V$ of~$Q$ associated with the eigenvalue $\rho \omega$, with~$\rho > 0$ and $\abs{\omega} = 1$. With $F_K(x) = \slSigma_K(x)/ (\rho \omega)^K$, Formula~\eqref{dumas-rrseq-v4:eq:RecursionFormula} of Lemma~\ref{dumas-rrseq-v4:lemma:RecursionFormula} becomes
\[
\rho \omega F_{K+1}(x) =  \sum_{b_1 < x_1} A_{b_1} V +  A_{x_1} F_K(\base x - x_1).
\]
If the sequence~$F_K(x)$ has a limit~$F(x)$, the previous equation gives in the limit
\[
\rho \omega F(x) = \sum_{b_1 < x_1} A_{b_1} V +  A_{x_1} F(\base x - x_1).
\]
According to Definition~\eqref{dumas-rrseq-v4:eq:VectorValuedSumDefinition}, we must have $F(1) = V$, because of $\slSigma_K(1) = \sum_{\length{w} = K}A_w V = Q^K V = (\rho \omega)^K V$.

More generally, if~$Q$ admits a generalized eigenvector associated with the eigenvalue~$\rho \omega$, there exists a free family~$\cV = (V^{(j)})_{0 \leq j < \nu}$ of vectors in~$\bC^d$, such that the subspace generated by~$\cV$ is left stable by~$Q$ and the matrix induced with respect to the basis~$\cV$ of this subspace is the usual Jordan cell with size~$\nu$
\[
J = J_{\nu,\rho \omega} = 
\left(\begin{array}{ccccc}
	 \rho\omega &1 & & & \\
              &\rho\omega&1 && \\
              &    & \ddots&\ddots & \\
              &    &       & & 1\\
              &    & & & \rho\omega
	 \end{array}\right).
\]
The same idea as in the case of an eigenvector leads us to consider the functional equation
\begin{equation}\label{dumas-rrseq-v4:eq:DilationEquation}
	F(x) J = \sum_{b_1 < x_1} A_{b_1} V +  A_{x_1} F(\base x - x_1).
\end{equation}
This time~$V$ and~$F(x)$ are matrix-valued, precisely their values are in~$\bC^{{[1,d]}\times {\left[0,\nu\right)}}$. Matrix~$V$ is made from the column vectors~$V^{(0)}$, $\ldots$, $V^{(\nu - 1)}$ in that order. Anew we must have $F(1) = V$. As functions~$\slSigma_K(x)$ and~$F_K(x)$ are defined on the segment~$[0,1]$, we are searching for solution of  previous Equation~\eqref{dumas-rrseq-v4:eq:DilationEquation} on~$[0,1]$. But we extend function~$F(x)$ by 
\begin{equation}\label{dumas-rrseq-v4:eq:BoundaryConditions}
	F(x) = 0 \quad \text{for $x \leq 0$},\qquad F(x) = V \quad \text{for $x \geq 1$}.
\end{equation}
With this choice, Equation~\eqref{dumas-rrseq-v4:eq:DilationEquation} rewrites as an homogeneous equation 
\begin{equation}\label{dumas-rrseq-v4:eq:HomogeneousDilationEquation}
	F(x)J = \sum_{0 \leq b < \base} A_b F(\base x - b).
\end{equation}

This equation is a \emph{dilation equation} or \emph{two-scale difference
equations}. Such equations are common in the definition of wavelets~\cite{DaLa91} and refinement schemes~\cite{MiPr89}, and have been studied in detail~\cite[Th.~2.2]{DaLa92}, \cite{Heil92,Rioul92}. They appear also in the study of   dynamical systems toy-examples~\cite{TaAnSu93,TaGiDo1998}. 

\medskip

Let us recall that a function~$f$ from the real line into a normed
space is H{\"o}lder continuous with exponent~$\alpha$ if it satisfies $\norm{f(y)-f(x)}\leq c |y-x|^{\alpha}$ for some constant~$c$. 

\begin{lemma}\label{dumas-rrseq-v4:lemma:TheoremAboutDilationEquations}
	For an eigenvalue with modulus $\rho > \rhojsr$, Equation~\eqref{dumas-rrseq-v4:eq:HomogeneousDilationEquation} has a unique continuous solution, which is H{\"o}lder with exponent~$\log_\base (\rho/ r)$ for every $r > \rhojsr$. 
\end{lemma}

\begin{proof}[Sketch of the proof]
	The existence and uniqueness of the solution results from the Banach fixed-point theorem~\cite[\S4]{DaLa91}. More concretely, there is a power of the fixed-point operator $\cL \Phi (x) = \sum_{0 \leq b < \base} A_b \Phi (\base x - b) J^{-1}$ which is a contraction mapping. The result about the H{\"o}lderian character of the solution is common in the study of dilation equations~\cite[p.~1039]{DaLa92}.
\end{proof}

\medskip

\noindent
The solution of the dilation equation is not always explicit, but it can always be computed by the algorithm known as \emph{ cascade algorithm}~\cite[\S6.5,
p.~206]{Daubechies92}, which enter into the family of \emph{interpolatory schemes}~\cite{DyLe02}. All pictures in this article are computed by this algorithm. We will use it immediately in  the next example.

\begin{example}[Dichopile algorithm, continued]\label{dumas-rrseq-v4:ex:Dichopile-3}
With the previous linear representation for the sequence~$\delta_n$, the matrix 
\[\displaystyle Q = A_0 + A_1 = \left( \begin {array}{cccccc} 0&0&1&0&0&0\\ \noalign{\medskip}1&1&0&1&0&0\\ \noalign{\medskip}1&0&1&1&0&0\\ \noalign{\medskip}0&1&0&0&0&0\\ \noalign{\medskip}1&0&0&0&2&1\\ \noalign{\medskip}1&2&0&0&0&1\end {array} \right). \]
admits the Jordan normal form $\Lambda = P^{-1} Q P$, with

\begin{equation}\label{dumas-rrseq-v4:eq:JordanDichopile}
\Lambda =  \left( \begin {array}{cccccc} 2&1&0&0&0&0\\ \noalign{\medskip}0&2&0&0
&0&0\\ \noalign{\medskip}0&0&1&1&0&0\\ \noalign{\medskip}0&0&0&1&0&0
\\ \noalign{\medskip}0&0&0&0&-1&0\\ \noalign{\medskip}0&0&0&0&0&0
\end {array} \right)
 ,\quad
P =  \frac{1}{12}\left( \begin {array}{cccccc} 0&2&0&6&-2&-6\\ \noalign{\medskip}0&4&0
&-6&2&0\\ \noalign{\medskip}0&4&0&6&2&0\\ \noalign{\medskip}0&2&0&-6&-
2&6\\ \noalign{\medskip}12&-16&6&15&1&0\\ \noalign{\medskip}0&10&-6&-
15&-1&6\end {array} \right). 
\end{equation}

Let us consider only the Jordan block relative to~$2$, in the upper left corner. It has size~$2$, and we have to find two vector-valued functions~$F^{0}(x)$ and~$F^{1}(x)$, which we rename~$f(x)$ and~$g(x)$ to shorten. The equation for~$f(x)$ is a homogeneous dilation equation. It is readily seen that it solves into $f_5(x) = x$ and $f_i(x) = 0$ for the other indices~$i$. (Obviously these expressions are valid only for $x$ in~$[0,1]$.) 

Let us consider more carefully the dilation equation for~$g(x)$. It is non-homogeneous and needs the knowledge of~$f(x)$ to be solved. It writes for $0 \leq x \leq 1$
\begin{align*}
 g_{1} \left( x \right) & = \frac{1}{2}\,g_{3} \left( 2\,x-1 \right) 
\\
 g_{2} \left( x \right) & = \frac{1}{2}\,g_{1} \left( 2\,x \right) +\frac{1}{2}\,g_{4} \left( 2\,x \right) +\frac{1}{2}\,g_{2} \left( 2\,x-1 \right) 
\\
 g_{3} \left( x \right) & = \frac{1}{2}\,g_{3} \left( 2\,x \right) +\frac{1}{2}\,g_{1} \left( 2\,x-1 \right) +\frac{1}{2}\,g_{4} \left( 2\,x-1 \right) 
\\
 g_{4} \left( x \right) & = \frac{1}{2}\,g_{2} \left( 2\,x \right) 
\\
 g_{5} \left( x \right) & = \frac{1}{2}\,g_{5} \left( 2\,x \right) +\frac{1}{2}\,g_{1} \left( 2\,x-1 \right) +\frac{1}{2}\,g_{5} \left( 2\,x-1 \right) +\frac{1}{2}\,g_{6} \left( 2\,x-1 \right) -\frac{1}{2}\, x 
\\
 g_{6} \left( x \right) & = \frac{1}{2}\,g_{1} \left( 2\,x \right) +\frac{1}{2}\,g_{2} \left( 2\,x \right) +\frac{1}{2}\,g_{6} \left( 2\,x \right) +\frac{1}{2}\,g_{2} \left( 2\,x-1 \right)  
\end{align*}
With~$g(0)$ and~$g(1)$ as input, the cascade algorithm computes first~$g(1/2)$, next~$g(1/4)$ and~$g(3/4)$, and in the end all the values~$g(i/2^K)$ with $0 \leq i \leq 2^K$ up to some depth~$K$, using the formul{\ae} above. In this way, we draw the pictures of Figure~\ref{dumas-rrseq-v4:fig:Dichopile-G-Functions}, and we see explicit expressions for~$g_1$, $g_2$, $g_3$, $g_4$, like
\[
g_1(x) = \left\{
\begin{array}{ll}
	0 & \text{if $0\leq x \leq 1/2$,}\\
	(x-1/2)/3 & \text{if $1/2 \leq x \leq 1$.}
\end{array}
\right.
\]
Such empirical results are not a proof, but the system has a unique continuous solution and a mere substitution establishes the formul{\ae}. Hence we know the four first components. Next we find elementarily 
\[
g_6(x) = \left\{
\begin{array}{ll}
	x/3 + 1/2 & \text{if $1/2 \leq x \leq 1$,}\\
	(k+2) x/3 + 1/(3\times 2^k) & \text{if $1/2^{k+1} \leq x \leq 1/2^k$ with $k \geq 1$.}
\end{array}
\right.
\]
However we do not have an explicit expression for~$g_5(x)$. We only know that it is well and  completely defined by the dilation equation 
\begin{equation*}
	g_5(x) = \frac{1}{2} g_5(2x) + \frac{1}{2} g_5(2x-1) + h(x) \qquad\text{for $0 \leq x \leq 1$,}
\end{equation*}
where~$h(x)$ is a known piecewise affine function, with the boundary conditions $g_5(0) = 0$, $g_5(1) = -4/3$. Moreover, it is H{\"o}lder continuous with exponent $\log_2 (2/r)$ for every $r > \rhojsr = 1$, that is with exponent $\varepsilon$ for every $\varepsilon < 1$.
\begin{figure}[t]
	\begin{center}
		\includegraphics[width=0.3\linewidth]{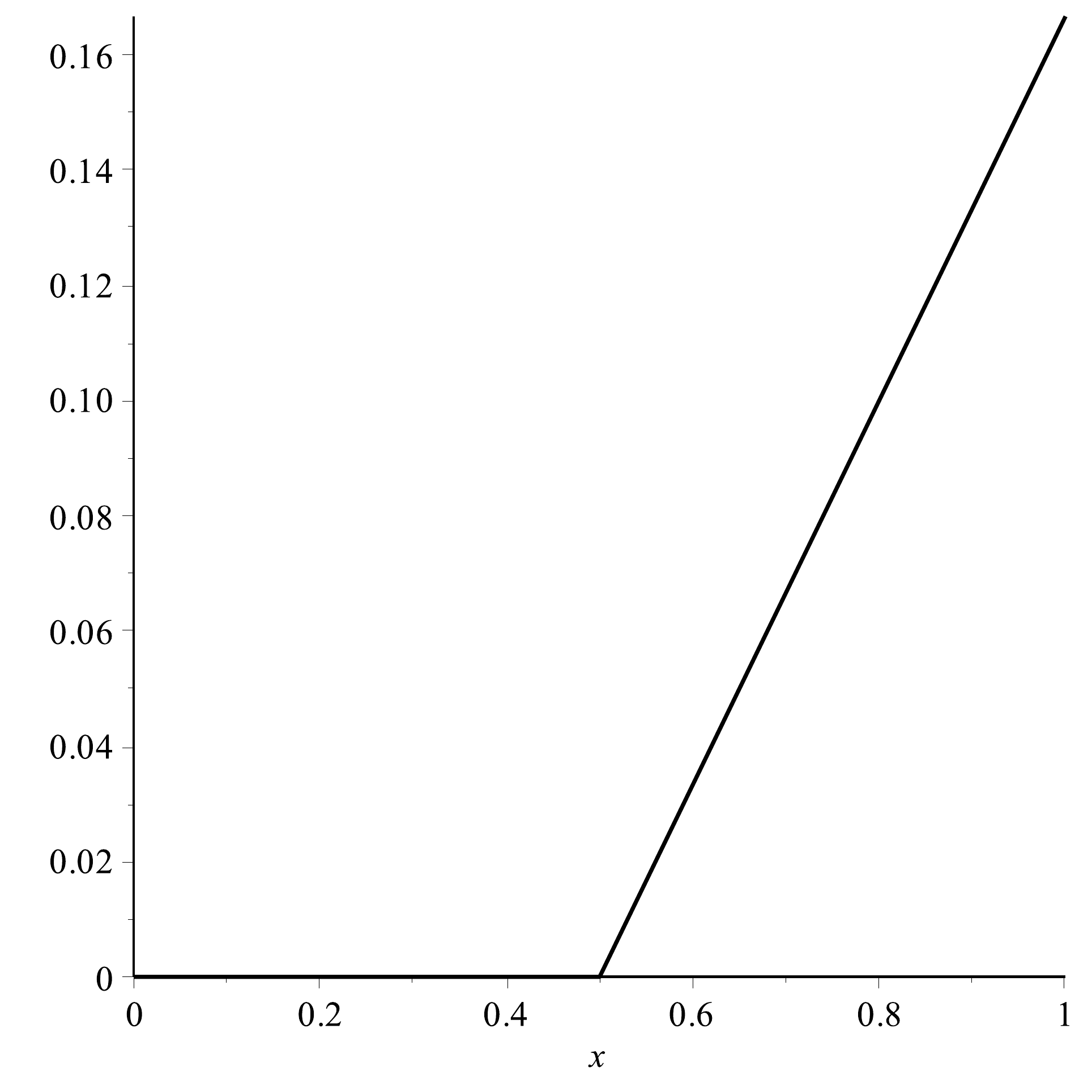}
		\hfil
		\includegraphics[width=0.3\linewidth]{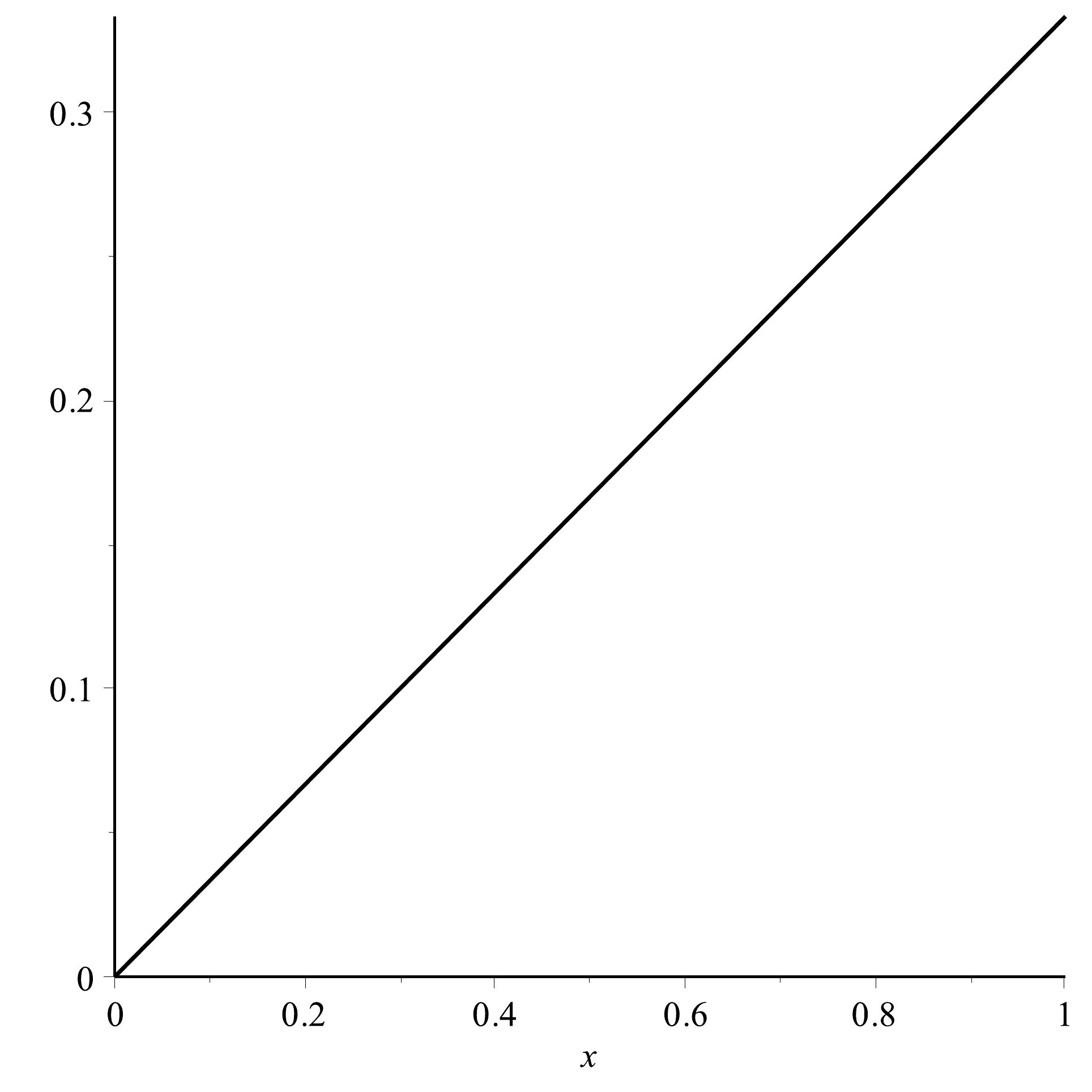}
		\hfil
		\includegraphics[width=0.3\linewidth]{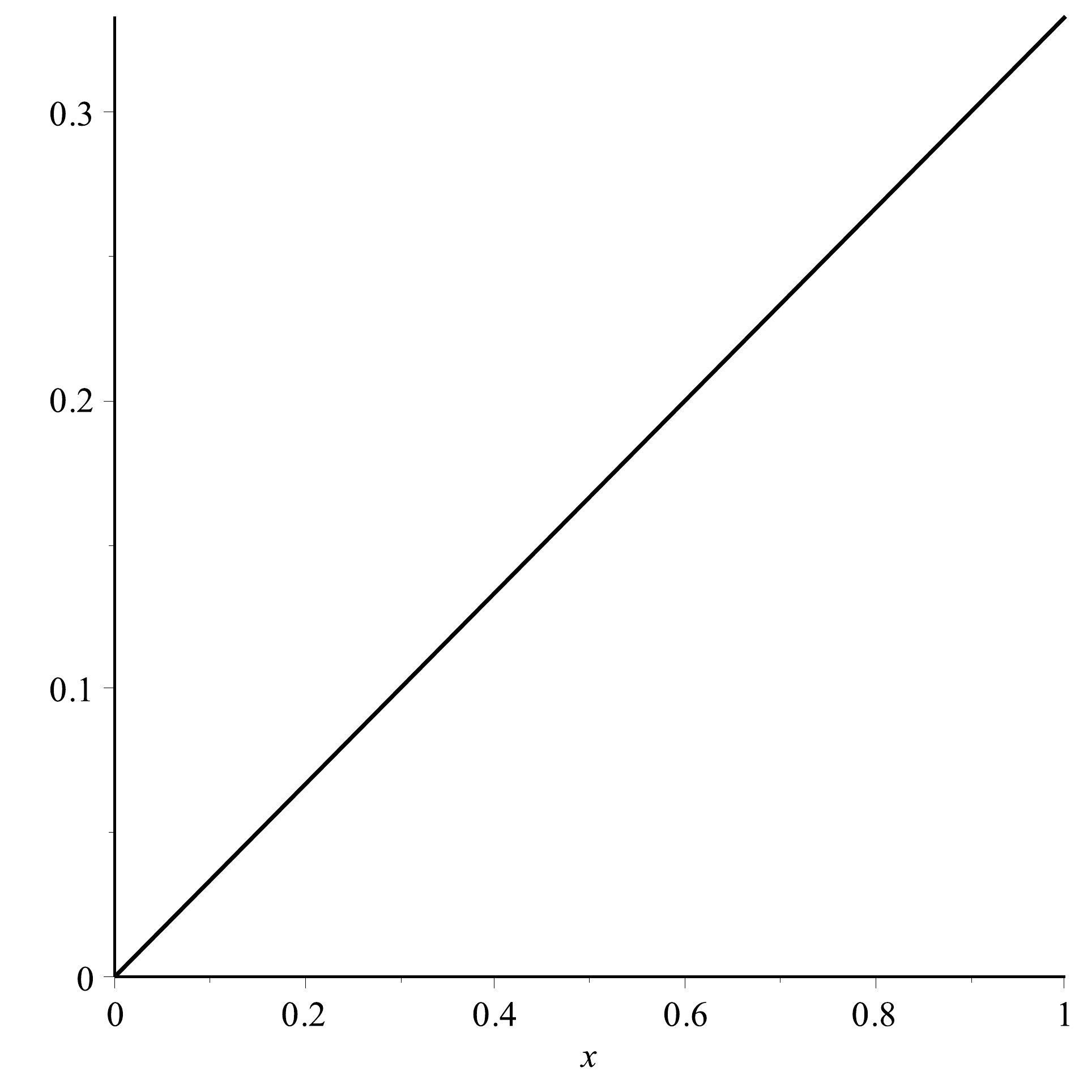}
		\\
		\includegraphics[width=0.3\linewidth]{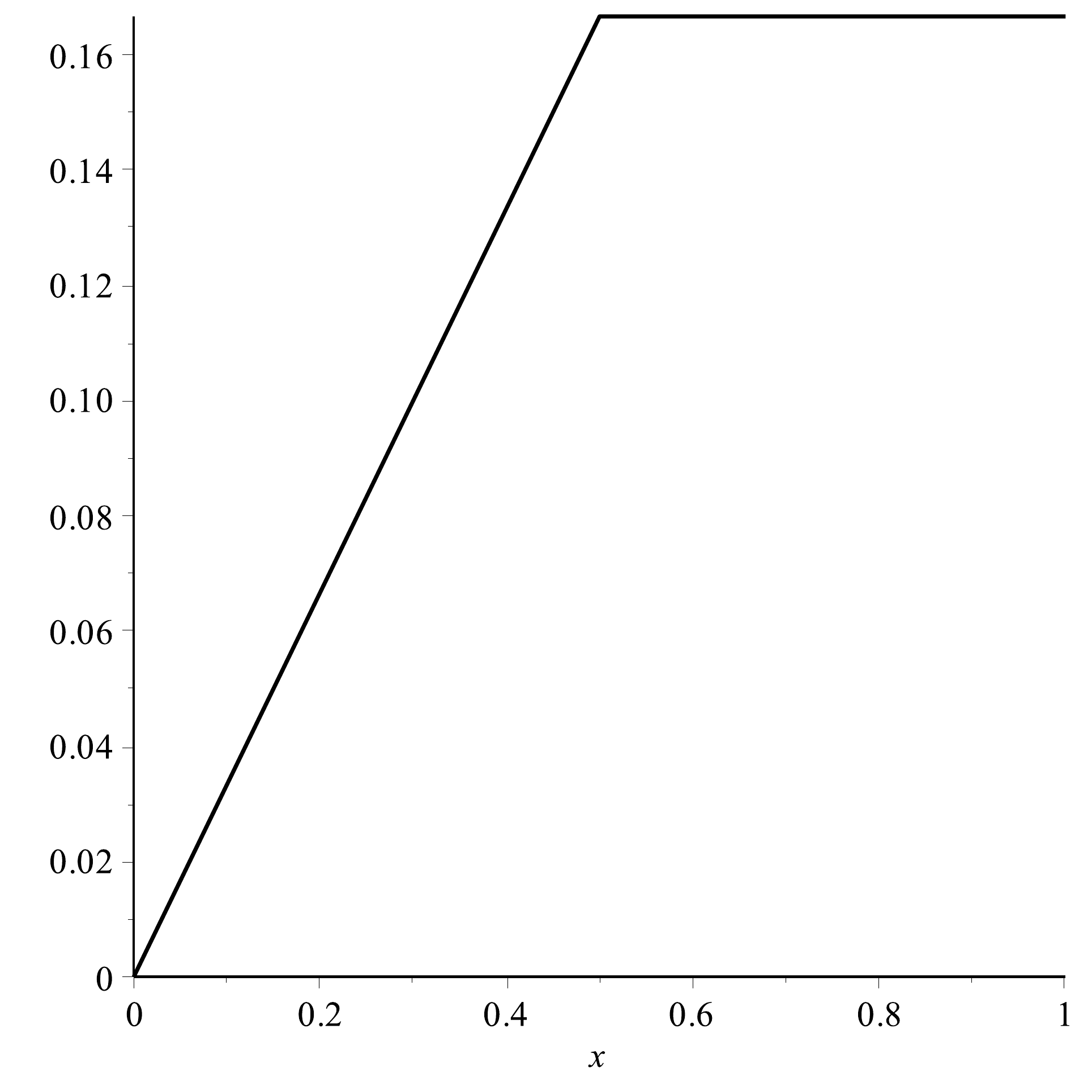}
		\hfil
		\includegraphics[width=0.3\linewidth]{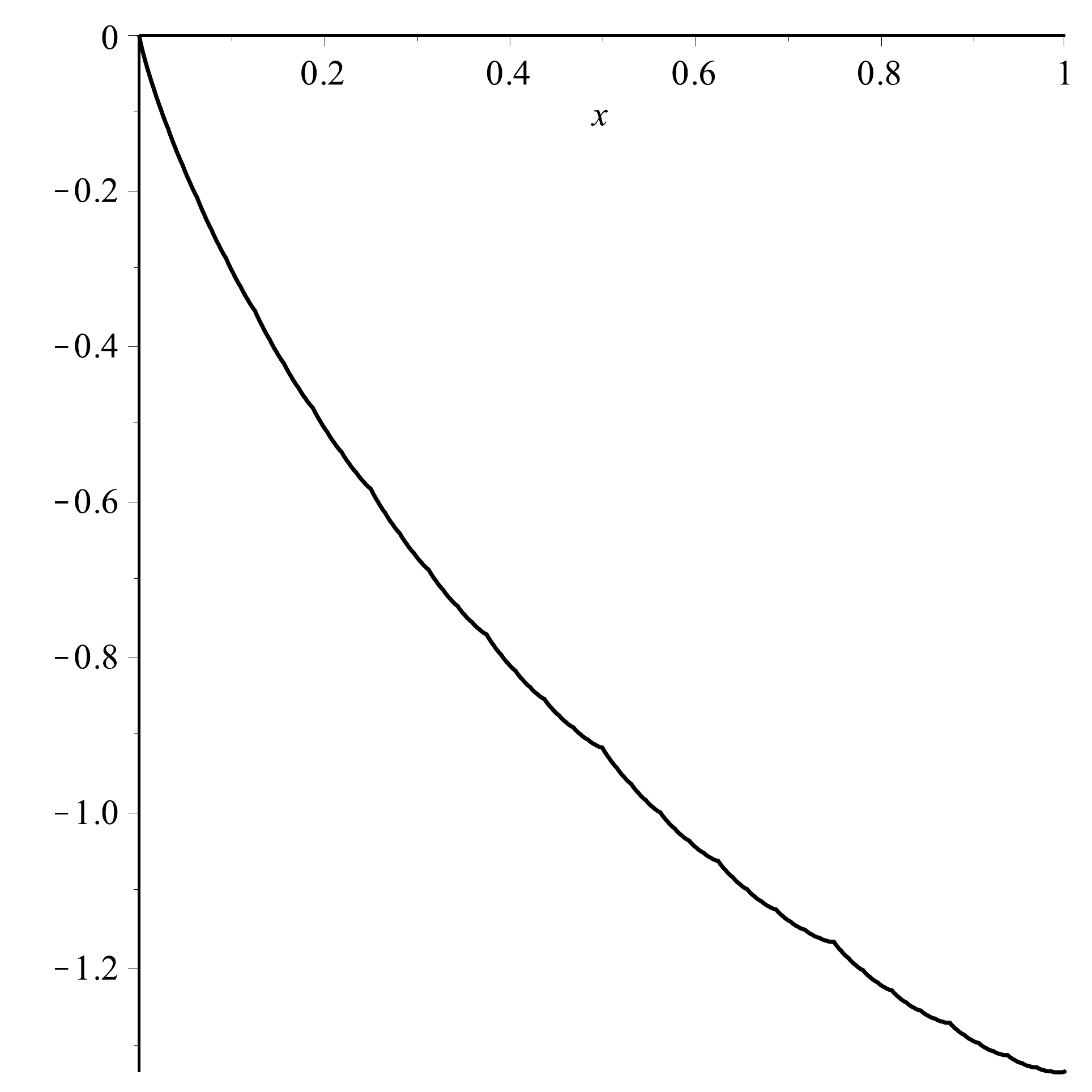}
		\hfil
		\includegraphics[width=0.3\linewidth]{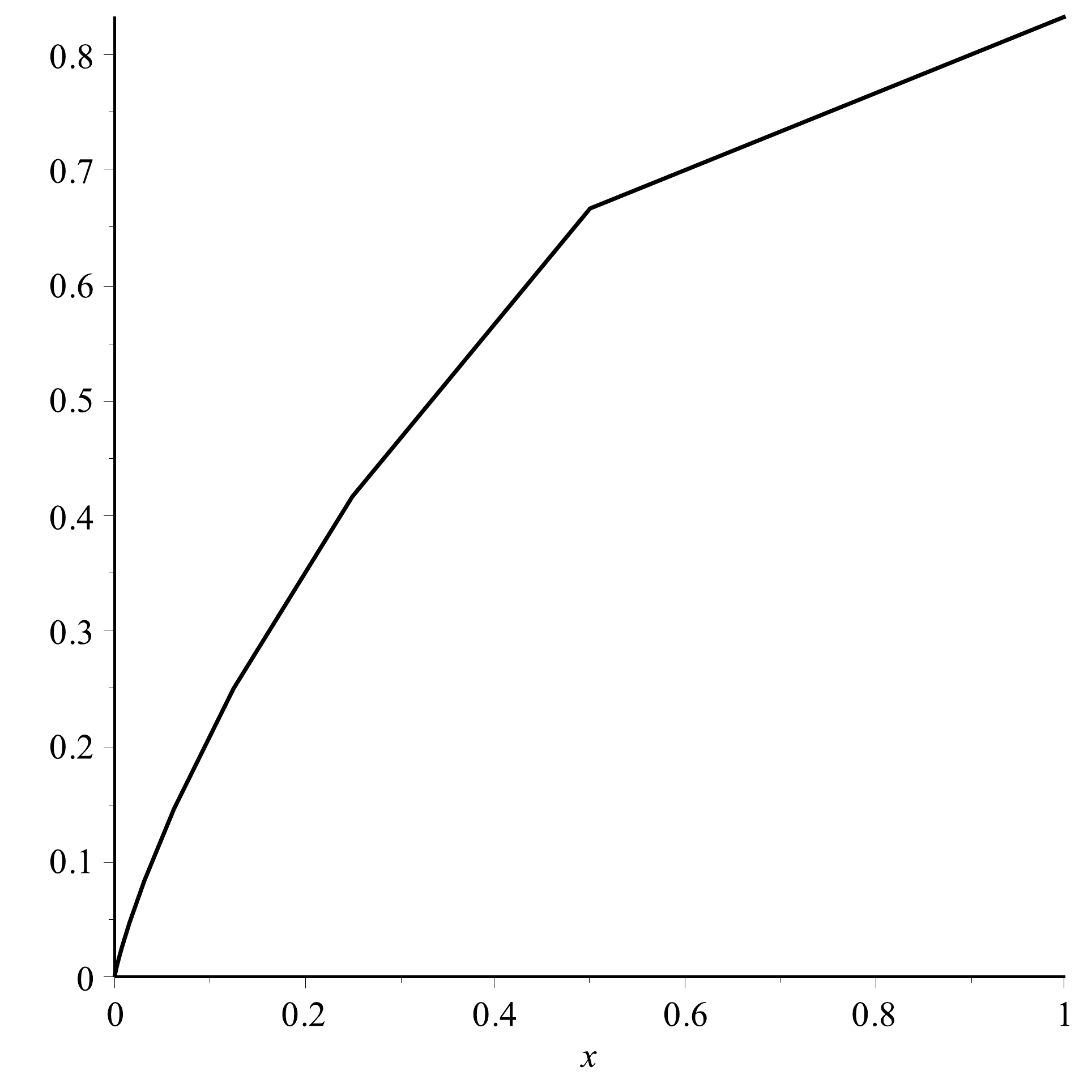}
	\end{center}
	\caption{\label{dumas-rrseq-v4:fig:Dichopile-G-Functions}
	The six components (from left to right, and from top to bottom) of the vector-valued function~$g(x) = F^1(x)$ for the dilation equation associated with the dichopile algorithm.
	}
\end{figure}
\end{example}

\section{Asymptotic expansion}\label{dumas-rrseq-v4:sec:AsymptoticExpansion}

\subsection{{Asymptotic expansion for the rational series}}
We proceed to a Jordan reduction of the matrix~$Q$, which is, let us call it to mind, the sum of the square matrices of the linear representation for the sequence~$u_n$. Next we expand the column vector~$C$ over the Jordan basis, so that we have to consider the sum~$\slSigma_K(x)$, defined by~\eqref{dumas-rrseq-v4:eq:VectorValuedSumDefinition}, relative to each vector of the Jordan basis. That is, using the same notations as in the previous section, we have to consider 
\[
	\slSigma_K(x) = \sum_{\begin{subarray}{c}\length{w} = K\\ \Bexp{w} \leq x\end{subarray}} A_w V^{(\nu -1)}.
\]
When $x$ goes from $0$ to $1$, the sum goes from nearly the null vector in~$\bC^d$ for $x = 0$ to the vector $Q^KV^{(\nu-1)}$ for $x = 1$. This last quantity evaluates to 
\begin{multline}
Q^KV^{(\nu-1)}=\binom K {\nu-1}(\rho\omega)^{K-\nu+1}V^{(0)}+
\binom K {\nu-2}(\rho\omega)^{K-\nu+2}V^{(1)}+
\dotsb\\+
\binom K 1 (\rho\omega)^{K-1}V^{(\nu-2)}+
(\rho\omega)^KV^{(\nu-1)},
\end{multline}
with the help of the relationship $Q V = V J$ (as in the previous section, matrix~$V$ has type $[1,d]\times \left[0,\nu\right)$ and its columns are~$V^{(0)}$, $\ldots$, $V^{(\nu - 1)}$). It is not too astonishing that the intermediate values may be described by the following expression
\begin{multline}\label{dumas-rrseq-v4:eq:asymptoticpartforJordancell}
A_K(x)=
\binom K {\nu-1}(\rho\omega)^{K-\nu+1}{F}^{(0)}(x)+
\binom K {\nu-2}(\rho\omega)^{K-\nu+2}{F}^{(1)}(x)+
\dotsb\\+
\binom K 1 (\rho\omega)^{K-1}{F}^{(\nu-2)}(x)+
(\rho\omega)^K{F}^{(\nu-1)}(x),
\end{multline}
at least asymptotically, since functions~$F^{(k)}(x)$ are obtained by a refinement scheme whose input are both values at the ends of the interval.

\begin{lemma}\label{dumas-rrseq-v4:lemma:AsymptoticExpansionJordanVector}
	For a Jordan vector~$V^{(\nu -1)}$ associated with an eigenvalue~$\rho \omega$, whose modulus~$\rho$ is larger than the joint spectral radius, that is $\rho > \rhojsr$, the expression~$A_K(x)$ above is an asymptotic expansion for the sum~$\slSigma_K(x)$. The error term is~$O(r^K)$ for every~$r$ between~$\rho$ and~$\rhojsr$, and it is uniform with respect to~$x$. 
\end{lemma}

\begin{proof}[Sketch of the proof]
	By substitution into the basic recursion Formula~\eqref{dumas-rrseq-v4:eq:RecursionFormula} of $\slSigma_K(x) = A_K (x) + \Delta_K (x)$, we find $\Delta_{K+1}(x) = A_{x_1} \Delta_K (\base x - x_1)$. Iterating this relationship, we bring out long products of matrices~$A_b$ whose norm is about $\rhojsr^K$ and consequently~$O(r^K)$.
\end{proof}

\begin{example}[Dichopile algorithm, continued]\label{dumas-rrseq-v4:ex:Dichopile-4}
Lemma~\ref{dumas-rrseq-v4:lemma:AsymptoticExpansionJordanVector} permits us to find an asymptotic expansion for~$S_K(x)$ as follows. We expand the column vector~$C$ of the linear representation~\eqref{dumas-rrseq-v4:eq:DichopileLinearRepresentation} onto the Jordan basis.
We find
\[
C = V_2^0 + V_2^1 + V_1^1 - 2V_{-1}^0 - V_0^0,
\]
with natural notations: $V_\lambda^j$ is a generalized eigenvector for the eigenvalue~$\lambda$.  
According to Lemma~\ref{dumas-rrseq-v4:lemma:ErrorTerm}, we have to take into account only the Jordan block relative to~$2$ in matrix~$\Lambda$ of Formula~\eqref{dumas-rrseq-v4:eq:JordanDichopile}, and consequently only 
$V_2^0 + V_2^1$. 

We apply Lemma~\ref{dumas-rrseq-v4:lemma:AsymptoticExpansionJordanVector} and specifically  Formula~\eqref{dumas-rrseq-v4:eq:asymptoticpartforJordancell} to obtain the contribution of each involved generalized eigenvector. 
With the notations of Example~\ref{dumas-rrseq-v4:ex:Dichopile-3}, the contribution of~$V_2^0$ to the sum~$\slSigma_K(x)$, relative to~$C$, is $2^K F^0(x) = 2^K f(x)$. 
Similarly, the contribution of~$V_2^1$ is $2^{K-1}K F^0(x) + 2^K F^1(x) = 2^{K-1}K f(x) + 2^K g(x)$, with $\nu = 1$.  So that we obtain
\begin{equation*}
	\slSigma_K(x) \mathop{=}_{K\to +\infty} 2^{K-1}(K+2) f(x) + 2^K g(x) + O(r^K)
\end{equation*}
for $1 < r < 2$. Henceforth the sum~$S_K(x) = L \slSigma_K(x)$, defined in Formula~\eqref{dumas-rrseq-v4:def:Sums}, satisfies
\begin{equation}\label{dumas-rrseq-v4:eq:AsymptoticExpansionForTheRationalSeries}
	S_K(x) \mathop{=}_{K\to +\infty} 2^{K-1}(K+2) x + 2^K g_5(x) + O(r^K).
\end{equation}

\end{example}

\vskip 2ex

\subsection{Asymptotic expansion for the radix-rational sequence}
Equation~\eqref{dumas-rrseq-v4:eq:BasicEquationSimplified} in Lemma~\ref{dumas-rrseq-v4:lemma:BasicEquation}, which relates~$s_N$ and~$S_{K+1}$, gives us the way to obtain the asymptotic expansion of the sequence~$s_N$, in other words the main result of this article.

\begin{figure}[t]
	\centering
		\includegraphics[width=0.5\linewidth]{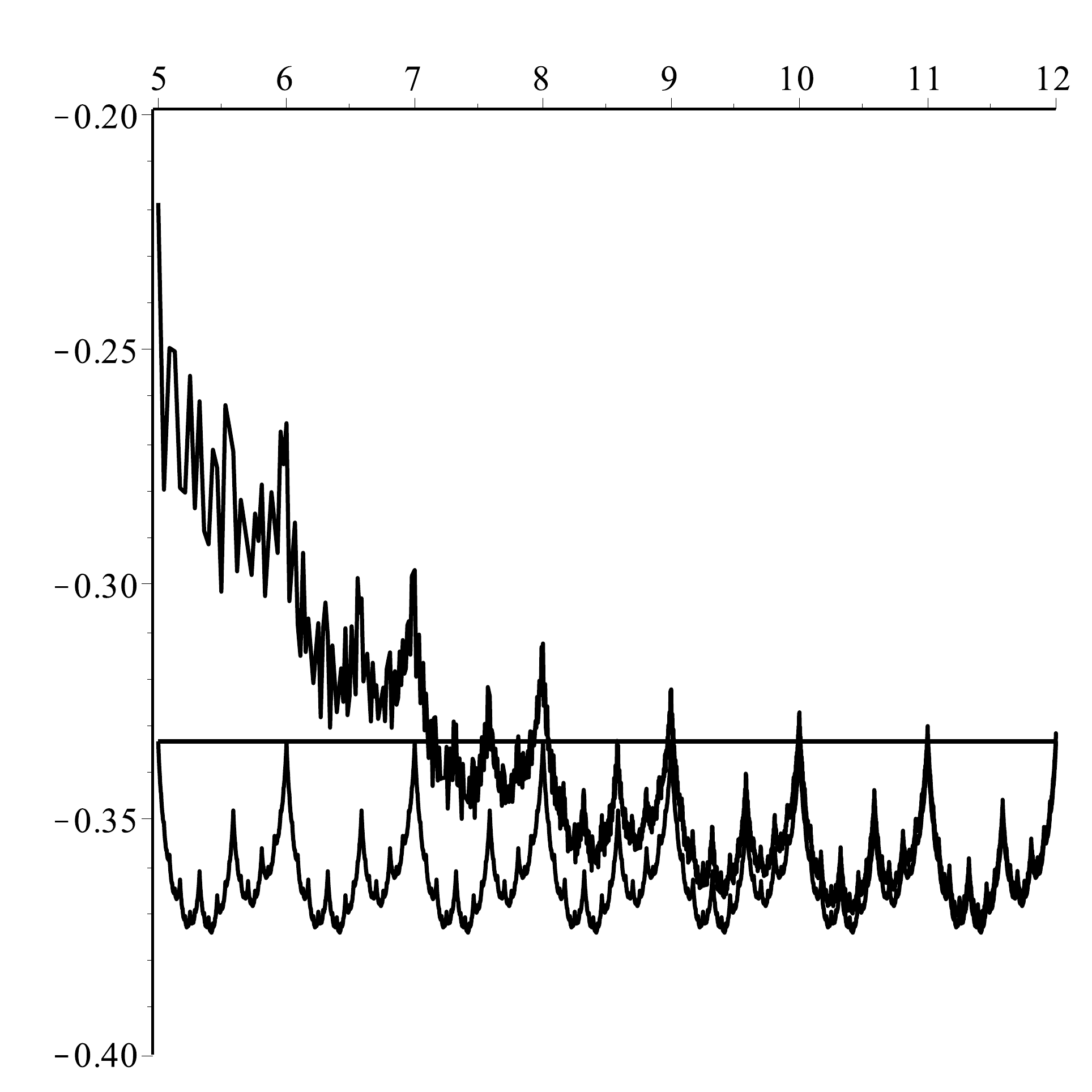}
	\caption{
	\label{dumas-rrseq-v4:fig:Dichopile-Phi-Function}
	Theory is closely akin to practice: for the dichopile algorithm, the normalized sequence $N^{-1} (f_N - N \log_2(N)/2) $ approaches the $1$-periodic function~$\Phi(t)$ as~$N$ goes towards infinity. Absciss{\ae} are in logarithmic scale for~$N$. The maximum value of the periodic function is~$-1/3$.
	}
\end{figure}

\begin{theorem}\label{dumas-rrseq-v4:thm:Theorem}
	Let~$s_N$ be a radix-rational sequence, whose the sequence of backward differences is defined by a linear representation~$L$, $(A_b)_{0 \leq b < \base}$, $C$,  insensitive to the leftmost zeroes. Then the sequence~$s_N$ admits an asymptotic expansion which is a sum of terms
\begin{equation}\label{dumas-rrseq-v4:eq:ASBasicTerm}
	N^{\log_\base \rho} \binom{\log_\base N}{m}	\times e^{i \vartheta \log_\base N} \times \Phi(\log_\base N),
\end{equation}
indexed by the eigenvalues~$\rho e^{i \vartheta}$ of~$Q = A_0 + \dotsb + A_{\base -1}$ whose modulus~$\rho$ is larger than the joint spectral radius~$\rhojsr$, with~$m$ a nonnegative integer and~$\Phi(t)$ a $1$-periodic function. The error term is~$O(N^{\log_\base r})$ for every $r > \rhojsr$. Moreover the $1$-periodic function~$\Phi(t)$
is H{\"o}lder continuous with exponent~$\log_\base (\rho / r)$ for every $r$ between~$\rho$ and~$\rhojsr$. 
\end{theorem}
It should be noted that the assumption of insensitivity is only a convenience that simplifies the computation. See~\cite{Dumas13} for a general version.

\begin{proof}[Sketch of the proof]
The error term is $O(r^{\log_\base N}) = O(N^{\log_\base r})$ for every~$r > \rhojsr$ according to Lemma~\ref{dumas-rrseq-v4:lemma:AsymptoticExpansionJordanVector}.  

A mere substitution in $S_{K+1}(x)$ (not in $S_K(x)$, however) translates the result about~$S_K(x)$ into the expected result about~$s_N$. First we write $\omega = e^{i \vartheta}$ in~\eqref{dumas-rrseq-v4:eq:asymptoticpartforJordancell} with~$\vartheta$ a real number. Next, for each term $\binom{K}{\ell} (\rho \omega)^{K-\ell} F(x)$, we change~$K$ into~$K + 1$ and~$x$ into $\base^{\sawtooth{t}-1}$. Moreover we rewrite $K = \lfloor \log_\base N \rfloor$ as $\log_\base N - \sawtooth{t}$ (implicitly $t = \log_\base N$) and we obtain the regular part of an asymptotic expansion for~$s_N$ as a sum of terms
\[
\binom{\log_\base N + 1 - \sawtooth{t}}{\ell} (\rho e^{i \vartheta})^{\log_\base N - (\ell - 1 + \sawtooth{t})}  F(\base^{\sawtooth{t} - 1}).
\]
With the Chu-Vandermonde identity and $\rho^{\log_\base N} = N^{\log_\base \rho}$, this rewrites
\begin{multline}\label{dumas-rrseq-v4:eq:ChuVandermonde}
N^{\log_\base \rho}	e^{i \vartheta \log_\base N} 
\rho^{-\ell + 1 - \sawtooth{t}} e^{-i \vartheta (\ell - 1 + \sawtooth{t})}
F(\base^{\sawtooth{t} - 1}) 
\times 
\sum_{m=0}^\ell \binom{\log_\base N}{m} \binom{1 - \sawtooth{t}}{\ell - m}
\\
=
\sum_{m=0}^\ell 
N^{\log_\base \rho}
\binom{\log_\base N}{m} \times
e^{i \vartheta \log_\base N} \times 
\rho^{-\ell + 1 - \sawtooth{t}} e^{-i \vartheta (\ell - 1 + \sawtooth{t})}
F(\base^{\sawtooth{t} - 1}) \binom{1 - \sawtooth{t}}{\ell - m}
.
\end{multline}
It turns out that the expansion is a linear combination of the more elementary terms~\eqref{dumas-rrseq-v4:eq:ASBasicTerm}.

	The proof is complete, except the point about the H{\"o}lderian character of~$\Phi(t)$. This property comes from Lemma~\ref{dumas-rrseq-v4:lemma:TheoremAboutDilationEquations}, but the use of the fractional part~$\sawtooth{t}$ needs a complement: we have to verify the continuity at integers. All in all, it amounts to the equality $QV = VJ$ with the notations of Formula~\eqref{dumas-rrseq-v4:eq:DilationEquation}. 
\end{proof}

\begin{example}[Dichopile algorithm, the end]\label{dumas-rrseq-v4:ex:Dichopile-5}
	Applying the substitution explained above, we obtain readily
\begin{equation}\label{dumas-rrseq-v4:eq:DichopileExpansion}
	f_N \mathop{=}_{N \to + \infty}
	\frac{1}{2} N \log_2 N + N \Phi(\log_2 N) + O(N^{\varepsilon}),
\end{equation}
for $0 < \varepsilon < 1$ and with a $1$-periodic function
\begin{equation}\label{dumas-rrseq-v4:eq:DichopilePeriodicFunction}
	\Phi(t) = \frac{3-\{t\}}{2} + 2^{1-\{t\}} g_5(2^{\{t\}-1}),
\end{equation}
which is H{\"o}lder with exponent $\varepsilon$ for every $\varepsilon < 1$.
Let us detail the change from one asymptotic expansion to the other. First, using Equation~\eqref{dumas-rrseq-v4:eq:AsymptoticExpansionForTheRationalSeries}, we write $S_{K+1}(x) = 2^K (K+3) x + 2^{K+1} g_5(x) + O(r^K)$. Next we replace~$K$ by $\log_2 N - \sawtooth{t}$ and~$x$ by~$2^{\sawtooth{t}-1}$ to obtain
\begin{multline*}
f_N \mathop{=}_{N \to + \infty} 
2^{\log_2 N - \sawtooth{t}} (\log_2N - \sawtooth{t} + 3) 2^{\sawtooth{t}-1} + 
2^{\log_2 N - \sawtooth{t} + 1} g_5(2^{\sawtooth{t}-1}) + O(N^{\log_2 r})
\\ =
\frac{1}{2} 2^{\log_2 N } \log_2 N + \frac{1}{2} 2^{\log_2 N} (3 - \sawtooth{t}) + 2^{\log_2 N} 2^{1-\sawtooth{t}} g_5(2^{\sawtooth{t}-1}) + O(N^{\log_2 r})
\end{multline*}
and the writing above appears. Figure~\ref{dumas-rrseq-v4:fig:Dichopile-Phi-Function} shows the comparison between the sequence $N^{-1} (f_N - N\log_2 (N)/2)$ and the periodic function~$\Phi(t)$. 
\end{example}

\section{Improvements}\label{dumas-rrseq-v4:sec:ExamplesAndImprovements}

\subsection{{Lazy approach}}
Our first improvement is in fact a worsening. Perhaps the reader finds that all this algebraic machinery is too  complicated for he wants only an asymptotic equivalent for the sequence under consideration. In that case, it suffices to simplify the process of computation.

\begin{example}[Dichopile algorithm, the true end]\label{dumas-rrseq-v4:ex:Dichopile-6}
Let us assume that we are only interested by an equivalent for the cost~$f_N$ of the dichopile algorithm. We compute the spectral radius of the matrix~$Q$ and we find it is~$2$. 
Using the maximum  absolute column sum norm, we find numerically for the length $T = 10$
\[
\max_{\lvert w \rvert = 10}\norm{A_w}_1^{1/10} \simeq 1.271 < 2.
\]
(As a matter of fact, the example is so simple that we find at hand the value~$\sqrt 3 \simeq 1.732$ for the length~$2$ and this is sufficient, but we try to be a little more generic.) We proceed as in Example~\ref{dumas-rrseq-v4:ex:Dichopile-4}, but we retain only the contribution of the vector~$V_2^{1}$, and moreover we take only the first term in~\eqref{dumas-rrseq-v4:eq:asymptoticpartforJordancell}, that is $\binom K {\nu-1}(\rho\omega)^{K-\nu+1}{F}^{(0)}(x)$. As often for the cost of an algorithm, the function~$F^{(0)}(x) = f(x)$ is explicit (Ex.~\ref{dumas-rrseq-v4:ex:Dichopile-3}) and we obtain
\begin{equation*}
	S_K(x) \mathop{\sim}_{K\to+\infty} 2^{K-1}K x,
\qquad
\text{henceforth}
\qquad
	f_N \mathop{\sim}_{N\to+\infty} \frac{1}{2} N \log_2 N.
\end{equation*}

\end{example}

\vskip 2ex

\subsection{{Improvement of the error term}}

It is slightly irritating that in using Theorem~\ref{dumas-rrseq-v4:thm:Theorem}, we have always to write 
\emph{the error term is~$O(N^{\log_\base r})$ for every $r > \rhojsr$.} It would be simpler if we could write \emph{the error term is~$O(N^{\log_\base \rhojsr})$.} This is not true in full generality, but there is a circumstance that guarantees this property. 
In the definition of the joint spectral radius~\eqref{dumas-rrseq-v4:def:JointSpectralRadius}, it can happen that there exists a subordinate norm and a length~$T$ such that the equality
\begin{equation}\label{dumas-rrseq-v4:eq:FinitenessProperty}
\rhojsr = \max_{|w|=T} \norm{A_w}^{1/T}
\end{equation}
takes place. In such a case, it is said \emph{the set of matrices has the finiteness property}~\cite{JuBl08}. We will said also that the linear representation has the finiteness property. 

The coordinate vector~$C$ of the representation decomposes onto the Jordan basis used to reduce the matrix~$Q$. The eigenvalues associated to the generalized eigenvectors which occur in this decomposition can be sorted into two sets: first the set~$\Lambda_{>}$ of eigenvalues of~$Q$ larger than the joint spectral radius, next its complementary part~$\Lambda_{\leq}$. The first set~$\Lambda_{>}$ provides the regular part of the asymptotic expansion. The other set~$\Lambda_{\leq}$ provides the error term. 

\begin{lemma}\label{dumas-rrseq-v4:lemma:ImprovmentErrorTerm}
If the linear representation has the finiteness property, the error term in the expansion announced by Theorem~\ref{dumas-rrseq-v4:thm:Theorem} writes
\begin{itemize}
	\item[--] $O(N^{\log_\base \rhojsr})$ if~$\Lambda_{\leq}$ does not contain the joint spectral radius~$\rhojsr$,
	\item[--] $O(N^{\log_\base \rhojsr} \log_\base ^ m N)$ if $\rhojsr$ is a member of~$\Lambda_{\leq}$ and $m$ is the maximal size of the Jordan cells associated to~$\rhojsr$ and involved in the decomposition of the coordinate vector~$C$ over the Jordan basis.
\end{itemize}
\end{lemma}

\begin{proof}[Sketch of the proof]
We deal with every eigenvalue not larger than~$\rhojsr$. We use the same notations as in Lemma~\ref{dumas-rrseq-v4:lemma:ErrorTerm}. In the first case $\rho < \rhojsr$, the proof of this proposition works with~$\rhojsr$ in place of~$r$. In the second case, let us assume for a while that we have $T = 1$ in~\eqref{dumas-rrseq-v4:eq:FinitenessProperty}. With the recursive formula~\eqref{dumas-rrseq-v4:eq:RecursionFormula} and $Q^K V^{\nu-1} = O(\rhojsr^K K^{\nu-1})$, we obtain
 $ \norm{\slSigma_{K+1}}_{\infty}\leq   c \rhojsr ^KK^{\nu-1}+ \rhojsr \norm{\slSigma_K}_{\infty}$ for some positive   constant~$c$. This recursion solves into $\norm{\slSigma_K}_{\infty} = O(\rhojsr^K K^{\nu})$. For the general case $T \geq 1$, we use the subsequences $\slSigma_{KT+s}(x)$ with $0 \leq s < T$. To have a bound at our disposal we employ~$m$ that is the maximal~$\nu$ encountered in the decomposition of~$C$. Substituting~$\log_\base N$ for~$K$, we arrive at the announced result.
\end{proof}

The next example enables us to see, with our eyes, that the order of the error term predicted by the above proposition is the right one.  

\begin{example}[Triangular tiling]\label{dumas-rrseq-v4:ex:geometricex}
This example is slightly different of the other examples in the article because first we do not consider a sequence but a rational series, second the linear representation is not insensitive to the leftmost zeroes. This last point is of no importance here because it was only useful in the change from a rational series to a radix-rational sequence.
For a real~$\vartheta$, we consider the rotation matrix
\[
R_{\vartheta}=\left(\begin{array}{cc}\cos{\vartheta} & -\sin{\vartheta}\\\sin{\vartheta}&\cos{\vartheta}\end{array}\right).
\]
In the following we use $V_0=E_1$ the first vector of the canonical
basis and $V_{\vartheta}=R_{\vartheta}V_0$. The example is based on the
linear representation, with radix~$2$, 
\[
L=\left(\begin{array}{cc}
1&0
	\end{array}\right),\qquad
A_0=R_{-\vartheta},\qquad
A_1=R_{\vartheta},\qquad
C=V_0.
\]
Because we use orthogonal matrices, the joint spectral radius is
$\rhojsr=1$. The matrix~$Q$ is the diagonal matrix $2 \cos \vartheta \, \matI_2$. 
We particularize the case to $\vartheta=\pi/3$, so that the only eigenvalue is $\rho=\rhojsr$ with $\nu =1$, and $m=1$. 

We consider the words $u=11$, $v=01$, and the rational
numbers~$x_{2K+2}$ and $y_{2K}$ whose binary expansions are
$(0.uv^K)_2$ and $(0.v^K)_2$ respectively. According to the functional
equation~\eqref{dumas-rrseq-v4:eq:RecursionFormula} satisfied by~$\slSigma_K(x)$, we have
\[ 
\slSigma_{2K+2}(x_{2K+2})=R_{-{\pi/3}}V_0+R_{{\pi/3}}\left(R_{-{\pi/3}}V_0+R_{{\pi/3}}\slSigma_{2K}(y_{2K})\right) 
=
V_{-{\pi/3}}+V_0+R_{2{\pi/3}}\slSigma_{2K}(y_{2K}),
\] 
\[
\slSigma_{2K}(y_{2K})=R_{-{\pi/3}}\left(R_{-{\pi/3}}V_0+R_{{\pi/3}}\slSigma_{2K-2}(y_{2K-2})\right)=
V_{-2{\pi/3}}+\slSigma_{2K-2}(y_{2K-2}).
\]
With $\slSigma_0(0)=V_0$, we obtain by induction  $\slSigma_{2K}(y_{2K})=KV_{-2{\pi/3}}+V_0$ and
\[
\slSigma_{2K+2}(x_{2K+2})=(K+1)V_0+V_{-{\pi/3}}+V_{2{\pi/3}}=(K+1)V_0.
\]
Hence the~$O(K^m\rhojsr^K)$ we have found in
Lemma~\ref{dumas-rrseq-v4:lemma:ImprovmentErrorTerm} is satisfying.
The orbit of the parameterized curve~$\slSigma_K(x)$ is illustrated in
Figure~\ref{dumas-rrseq-v4:fig:triangular_cells}. 
\end{example}

\begin{figure}[t]
\begin{center}
\includegraphics[width=0.3\linewidth]{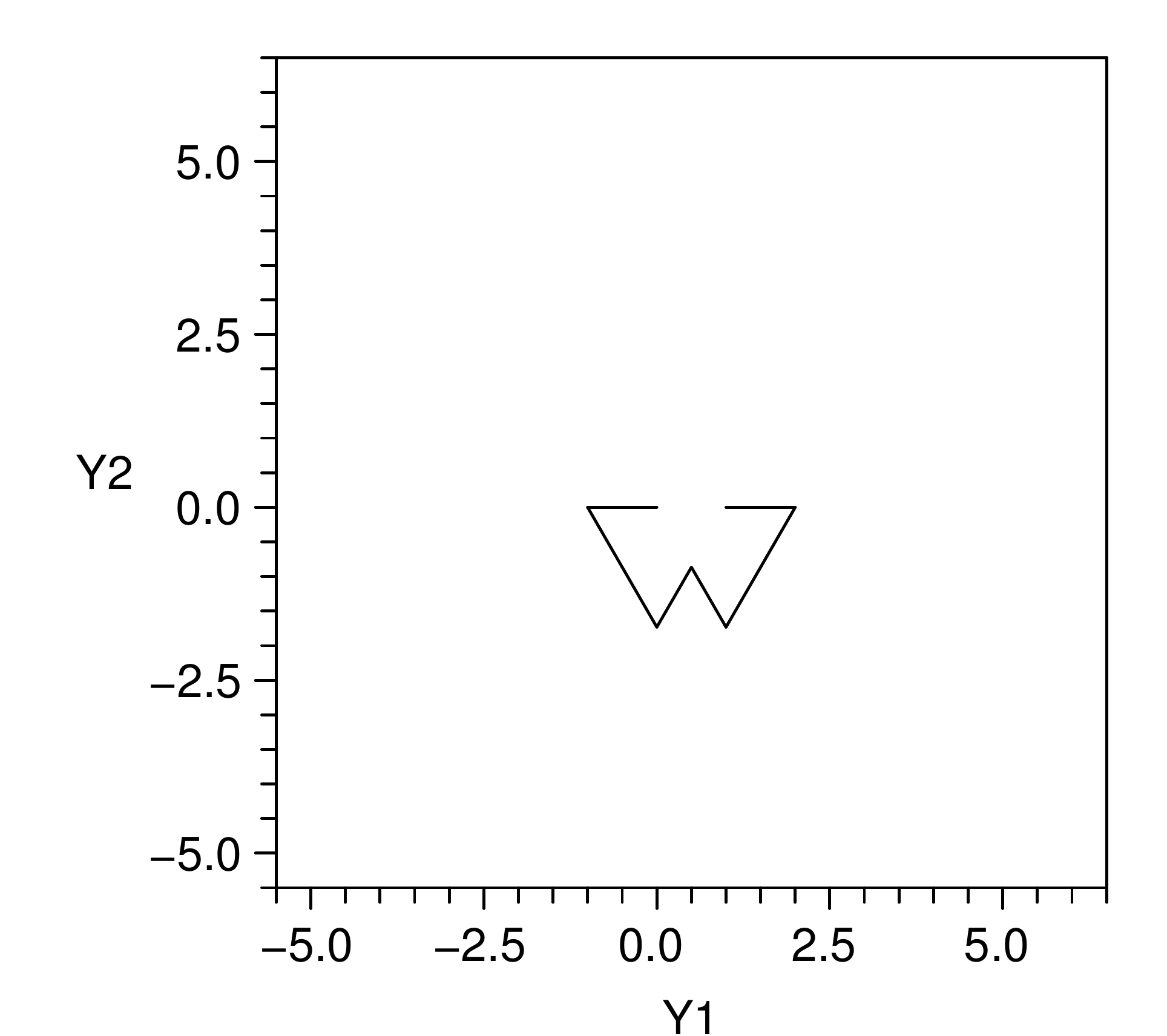}
\hfil
\includegraphics[width=0.3\linewidth]{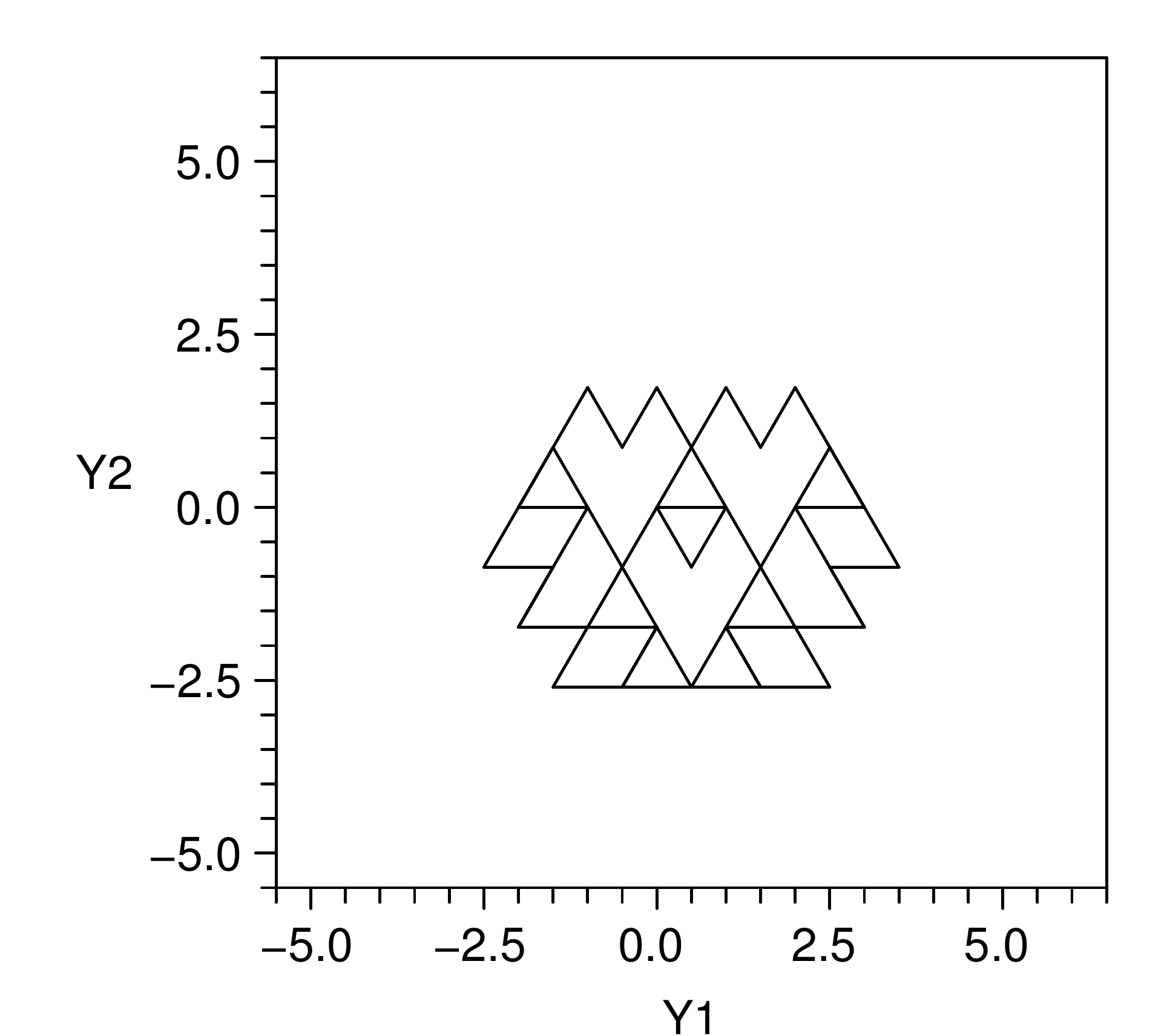}
\hfil
\includegraphics[width=0.3\linewidth]{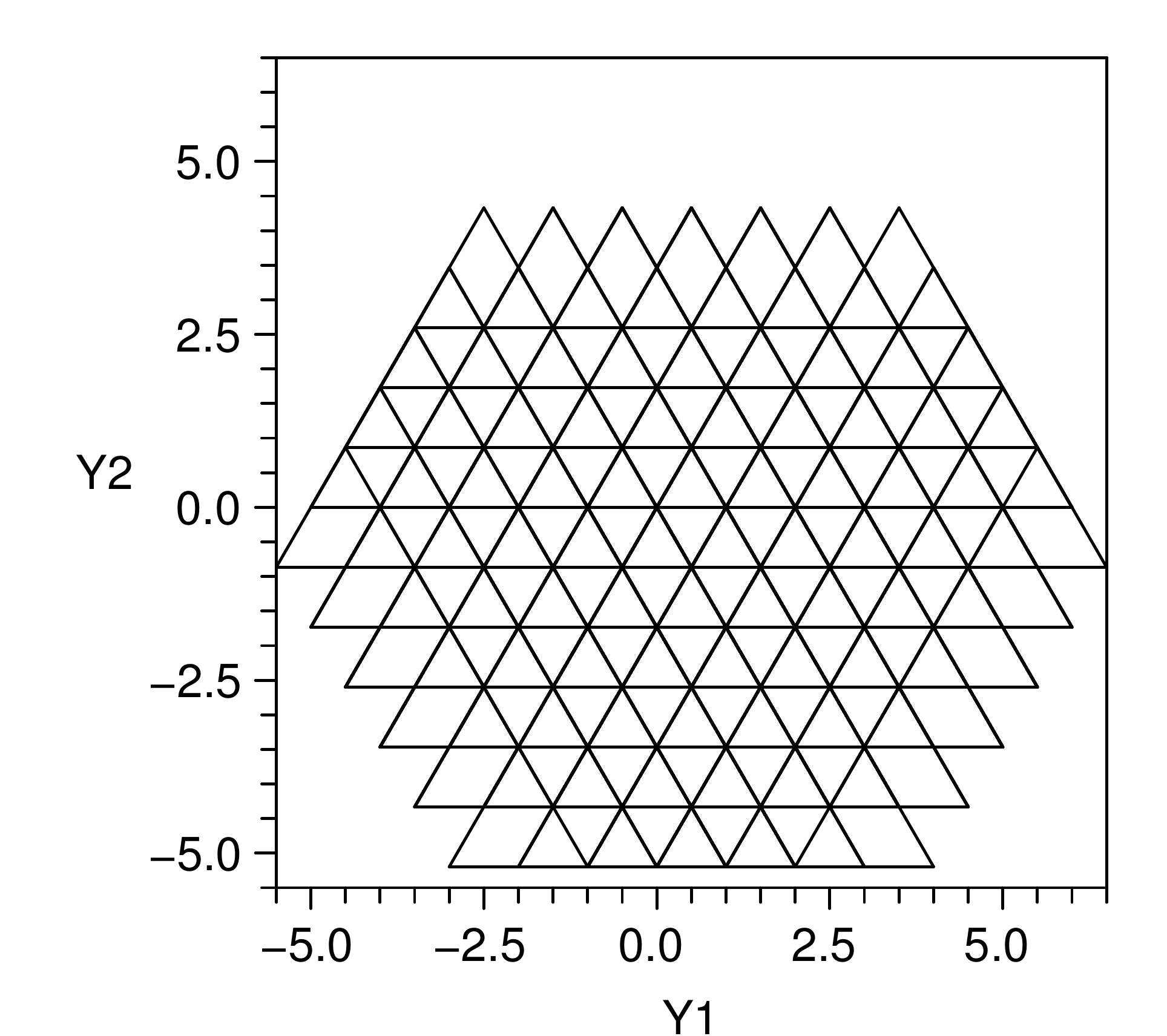}
\end{center}
\caption{\label{dumas-rrseq-v4:fig:triangular_cells}The image of the parameterized curve~$\slSigma_K(x)$ from Ex.~\ref{dumas-rrseq-v4:ex:geometricex} grows linearly with~$K$. Here we see $\slSigma_3$, $\slSigma_6$, and $\slSigma_{12}$.
}
\end{figure}

\subsection{Algorithm}

To summarize the obtained results and clarify the process of computation we write it as Algorithm {\tt LRtoAE} (for Linear Representation to Asymptotic Expansion). To tell the truth, it is a \emph{ pseudo\/}-algorithm (Algorithm~\ref{dumas-rrseq-v4:algo:AlgoBeginning}, p.~\pageref{dumas-rrseq-v4:algo:AlgoBeginning}%
). For example in Line~\ref{lrtoaeline3}, it can be difficult to compute the joint spectral radius~\cite{TsBl97}. Perhaps we are obliged to content ourselves with a real number~$r$ slightly larger than~$\rhojsr$ as in Example~\ref{dumas-rrseq-v4:ex:Dichopile-6}. In that case the question about the finiteness property (Line~\ref{lrtoaeline11}) is no longer meaningful. In the same way, it is likely that in Line~\ref{lrtoaeline14} we cannot solve explicitly the dilation equation. This is not an obstacle that prevents us from writing the expansion. Then it provides us with the qualitative behaviour of the sequence, but the cost of computing the asymptotic expansion for a given integer is almost the same as computing the value of the sequence for that integer.

\begingroup
\begin{algorithm}\label{dumas-rrseq-v4:algo:AlgoBeginning}\small
\caption{\label{dumas-rrseq-v4:algoforrationalsequencespiece1} {\tt LRtoAE}
}
\SetKwInOut{KwInput}{Input}
\SetKwInOut{KwOutput}{Output}
\SetAlgoLined
\SetArgSty{rm}
  \KwInput{A linear representation $L$, $(A_b)_{0 \leq b < \base}$, $C$ insensitive to the leftmost zeroes for the sequence of backward differences $u_n = \nabla s_n$.}
  \KwOutput{An asymptotic expansion of the sequence $s_N$ with respect to the scale
  $N^{\alpha}\binom {\log_\base N} \ell$, $\alpha\in\bR$,
  $\ell\in\mathbb N_{\geq 0}$
  .}
%
  $Q:=\sum_{0\leq b <\base}A_b$\;
  compute a Jordan basis $\cV$ for the matrix $Q$\nllabel{lrtoaeline2}\;
  compute the joint spectral radius $\rhojsr$ of the linear representation\nllabel{lrtoaeline3}\;
  \eIf{the linear representation\ has the finiteness property\nllabel{lrtoaeline11}}
      {$ r :=\rhojsr$\nllabel{lrtoaeline12}}
      {$ r :=$ any number between $\rhojsr$ and the infimum of
      the modulus of eigenvalues of~$Q$ greater than~$\rhojsr$}
  expand the column vector~$C$ of the linear representation over the
  Jordan basis, as $C=\sum_{V\in\cV}\gamma_VV$\;
  $\cV_{>}:=$ the set of generalized eigenvectors in~$\cV$ such that $\gamma_V\neq0$ and the associated eigenvalue~$\rho\omega$ has a modulus $\rho> r $\;
%
  \For{each vector~$V$ in~$\cV_{>}$} 
      {$\rho\omega:=$ the eigenvalue associated to~$V$\;
       $\nu:=$ the size of the Jordan cell~$J$ associated to~$V$\;
       compute the solution~$F(x)$ in $\bC^{\left[0,\nu\right)} \times \bC^{[1,d]}$ of the dilation system
$\displaystyle
F(x) J = \sum
A_b F(\base x -b)
$ \\[-0.5ex]
with the boundary conditions $F(x) = 0$ for $x \leq 0$ and $F(x) = V $ for $x \geq 1$\nllabel{lrtoaeline14}\;
  write down the expansion, with $F^{(-1)}(x) = 0$,
  \begin{equation*} 
   E_{V,N}:=
\sum_{\ell=0}^{\nu-1}
\binom{K+1}{\ell}(\rho\omega)^{K+1-\ell}\left[
\rho\omega  F^{(\nu-\ell-1)}(\base^{ \sawtooth{t} -1})+ F^{(\nu-\ell-2)}(\base^{ \sawtooth{t} -1})
\right]
  \end{equation*} 
  substitute $\log_\base N -  \sawtooth{t} $ for~$K$ in $ E_{V,N}$ (implicitly $t = \log_\base N$)\;
  expand the binomial coefficients $\binom{\log_\base N + 1 -  \sawtooth{t} }{\ell}$ by the Chu-Vandermonde formula\;
  collect~$ E_{V,N}$ according to the binomial coefficients $\binom{\log_\base N}{m}$
}
  $ E_N:=\sum_{V\in\cV_{>}}\gamma_{V} E_{V,N}$\;
  collect $E_N$ according to the powers of~$N$ and next  to the index~$m$ of the binomial coefficients  to obtain
\begin{equation*}
	E_N = \sum_{\rho,\vartheta, m} N^{\log_\base \rho} \binom{\log_\base N}{ m} e^{i \vartheta \log_\base N} \Phi_{\rho,\vartheta, m}(\log_\base N)
\end{equation*}
where $\Phi_{\rho,\vartheta, m}(t)$ is $1$-periodic and H{\"o}lder~$\log_\base (\rho/\rhojsr)$\;
  $\operatorname{error~term}_N:=O(N^{\log_\base  r })$\;
  \If{the linear representation has the finiteness property\nllabel{lrtoaeline5}}
     {$\cV_=:=$ the set of generalized eigenvectors in~$\cV$ such that
     $\gamma_V\neq0$ and the associated eigenvalue~$\rho\omega$ has a
     modulus $\rho=\rhojsr$\;
     $\nu_{\max}:=$ the maximal size of the Jordan cell for the vectors~$V$ in~$\cV_{=}$\;
     $\operatorname{error~term}_N:=O(N^{\log_\base \rhojsr}\log^{\nu_{\max}}_\base N)$\nllabel{lrtoaeline13}}
  \Return{
     \[\displaystyle s_N\mathop{=}_{N\to+\infty} L  E_N+\operatorname{error~term}_N\]
  }
\end{algorithm}
\endgroup

\subsection{More improvements}

Certainly more improvements are possible. For example, in our theorem, we consider a H{\"o}lder exponent. It must be understood that this exponent is a global lower bound, which means it is valid uniformly in the whole interval of reference. 
A deeper approach is presented in~\cite[Th.~4.2, p.~1054]{DaLa92} where the idea of a local H{\"o}lder exponent is described and studied in the framework of dilation equations. 

In~\cite{Tenenbaum97}, Tenenbaum shows that some periodic functions are nowhere differentiable. It would be a misunderstanding of Tenebaum's article to think that this is the general case. As a matter of fact, Tenebaum assumes that the H{\"o}lder exponent is an upper bound and not a lower bound (in mathematical notations $\Omega(\abs{h}^\alpha)$ and not $O(\abs{h}^\alpha)$). Below, we consider a probability distribution function. According to the Lebesgue differentiation theorem for monotone functions, it is almost everywhere differentiable. A good bibliography about nowhere differentiable functions and their history can be found in~\cite{AlKa2011,Lagarias12}. 

\begin{figure}[t]
\begin{center}
  \includegraphics[width=0.45\linewidth]{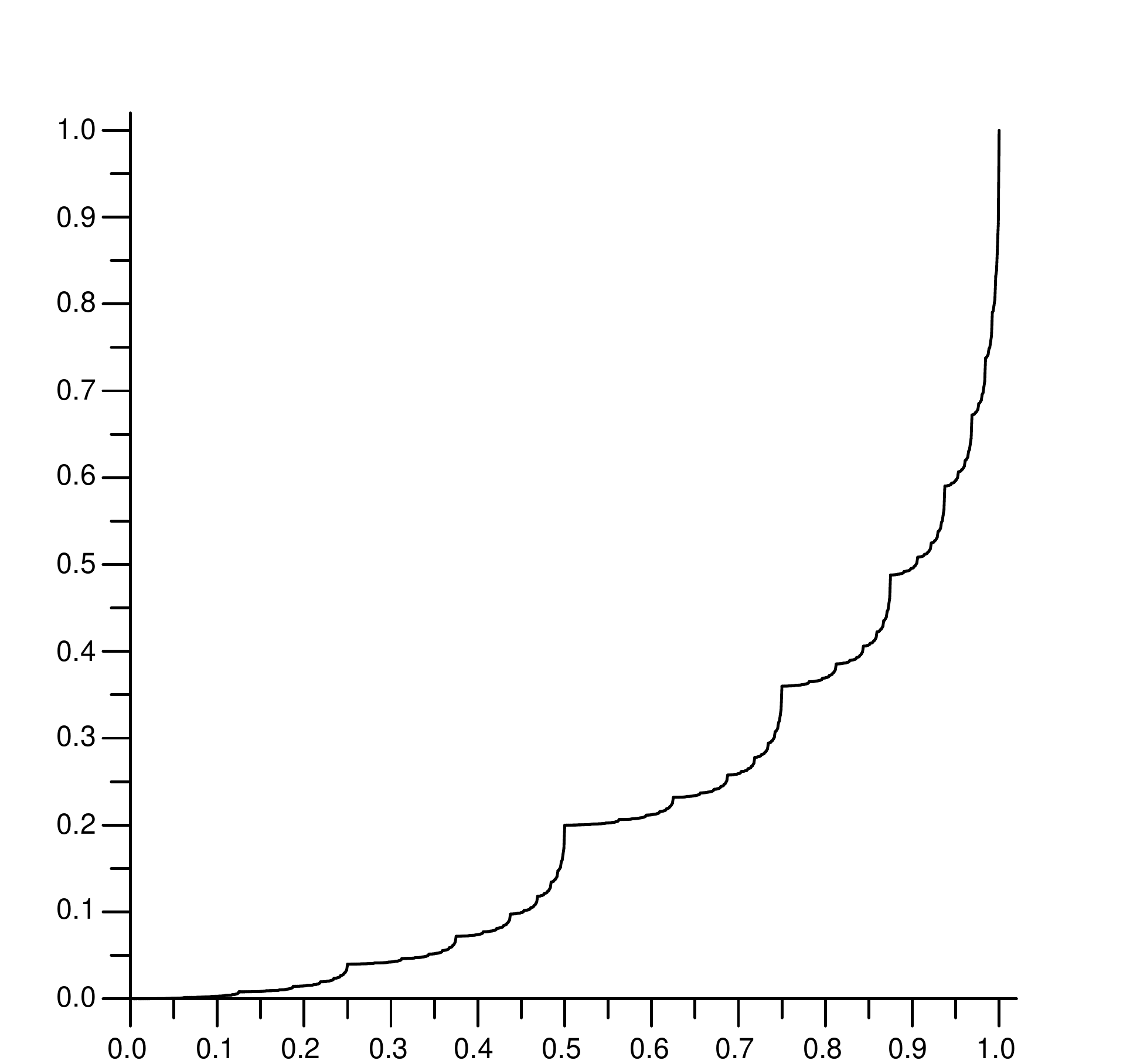}
\end{center}
\caption{\label{dumas-rrseq-v4:fig:Billingsley}Limit probability
  distributions which come from a Bernoulli process described in
  Ex.~\ref{dumas-rrseq-v4:ex:Billingsley}.}
\end{figure}

\begin{example}[Biased coin distribution function]\label{dumas-rrseq-v4:ex:Billingsley}
Billingsley~\cite[Ex.~31.1, p.~407]{Billingsley95} studied the random variable   $X=\sum_{n\geq0}X_n/2^n$ where~$X_n$ is the result of a coin tossing   with probabilities~$p_0$ and~$p_1$ for $X_n=0$ and $X_n=1$   respectively. This defines a rational series with dimension~$1$,   radix~$2$ and a linear representation  $L = (1)$, $A_0 = (p_0)$, $A_1 = (p_1)$, $C = (1)$ with $0<p_0,p_1<1$,   $p_0+p_1=1$. From the standpoint of number theory, the associated sequence is a completely $2$-multiplicative function~\cite[Def.~8.1.5]{DrGr10}. We have $Q= (1)$,   $\rho=1$, $V=(1)$,  $\rhojsr=\max(p_0,p_1)$. The distribution function~$F$   is the solution of the dilation equation~\eqref{dumas-rrseq-v4:eq:DilationEquation}  and   it is H\"older with exponent $\log_2(1/\rhojsr)$. We illustrate   the example with   $p_0=1/5$, $p_1=4/5$ and the exponent is $\log_2(4/5)\simeq   0.322$, which provides us with Figure~\ref{dumas-rrseq-v4:fig:Billingsley}, already sketched out in~\cite[p.~268--269]{LoUl34} or~\cite{Kawamura2011}, where the distribution function is quoted as {\rm Lebesgue's singular function}. A similar pictures appears also in~\cite[Fig.~6]{TaGiDo1998} about baker-type maps and in~\cite{Bedford1989}, which is a gentle introduction to the idea of box dimension.  

The provided exponent is the best uniform exponent. The positive character of the representation permits us to show (see~\cite{DuLiWa07} for an example) that, assuming $p_0\leq p_1$, at every dyadic point   the best local H\"older exponent is $\log_2(1/p_0)$ on the right-hand side   and $\log_2(1/\rhojsr)$ on the left-hand side. Except in the case   $p_0=p_1=1/2$, this gives $\log_2(1/p_0)>1$ and this explains the right-hand sided    horizontal tangents  in the   pictures. See\cite{Kawamura2011} for details about the derivative and a bibliography.
\end{example}


A recurrent question is about the bounds of the periodic functions. It is dealt with for example in~\cite{BrErMo83,Coquet83} and more recently in~\cite{AlKa2006,Kruppel2007}.

An issue which does not seem to have been tackled yet concerns the symmetries of the solutions of dilations equations.  It appears slightly in~\cite{Lagarias12} with the symmetry $\tau(1-x) = \tau(x)$ of Takagi function. We examplify it with the Rudin-Shapiro sequence.

\begin{figure}[t]
\begin{center}
\includegraphics[width=0.45\linewidth]{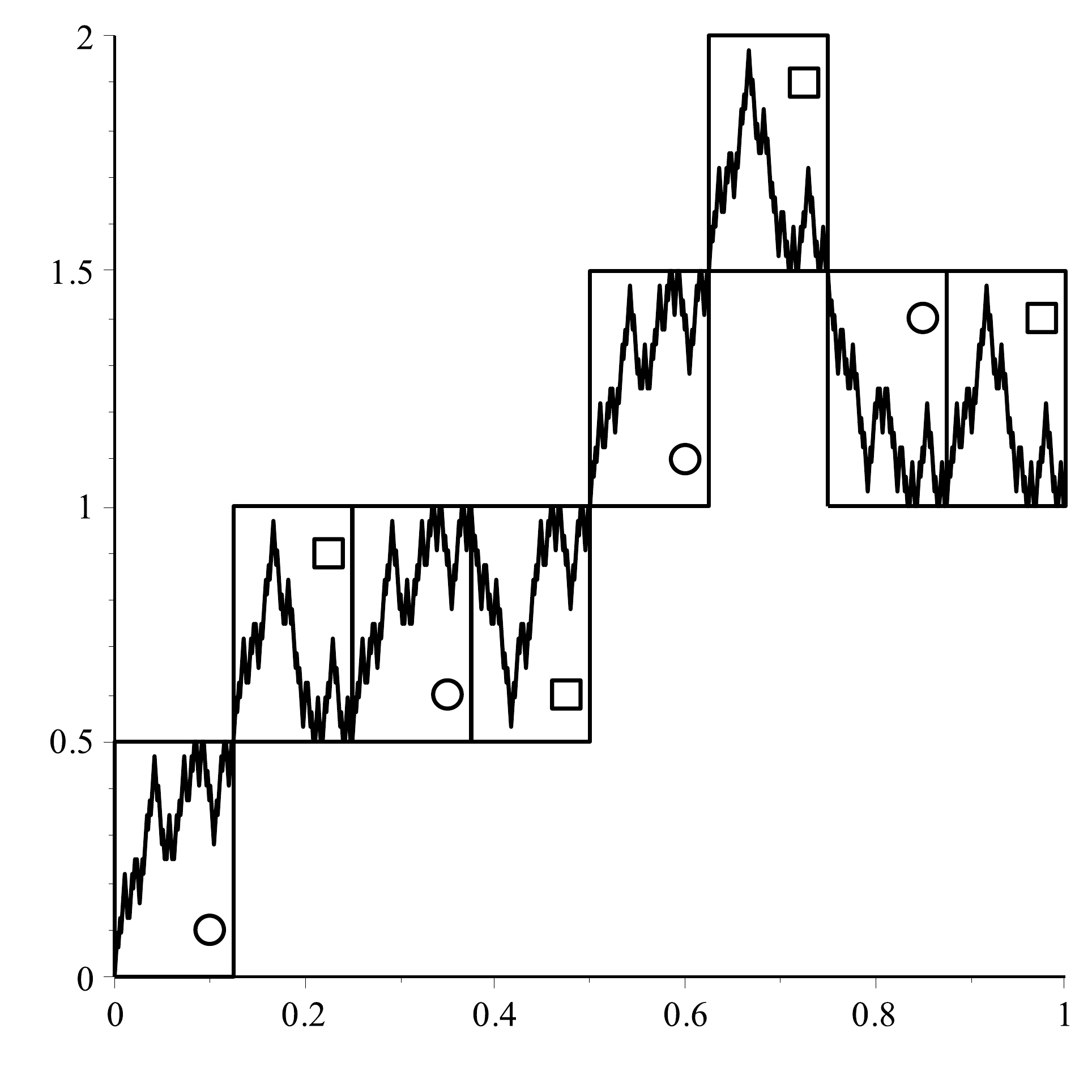}
\hfil
\includegraphics[width=0.45\linewidth]{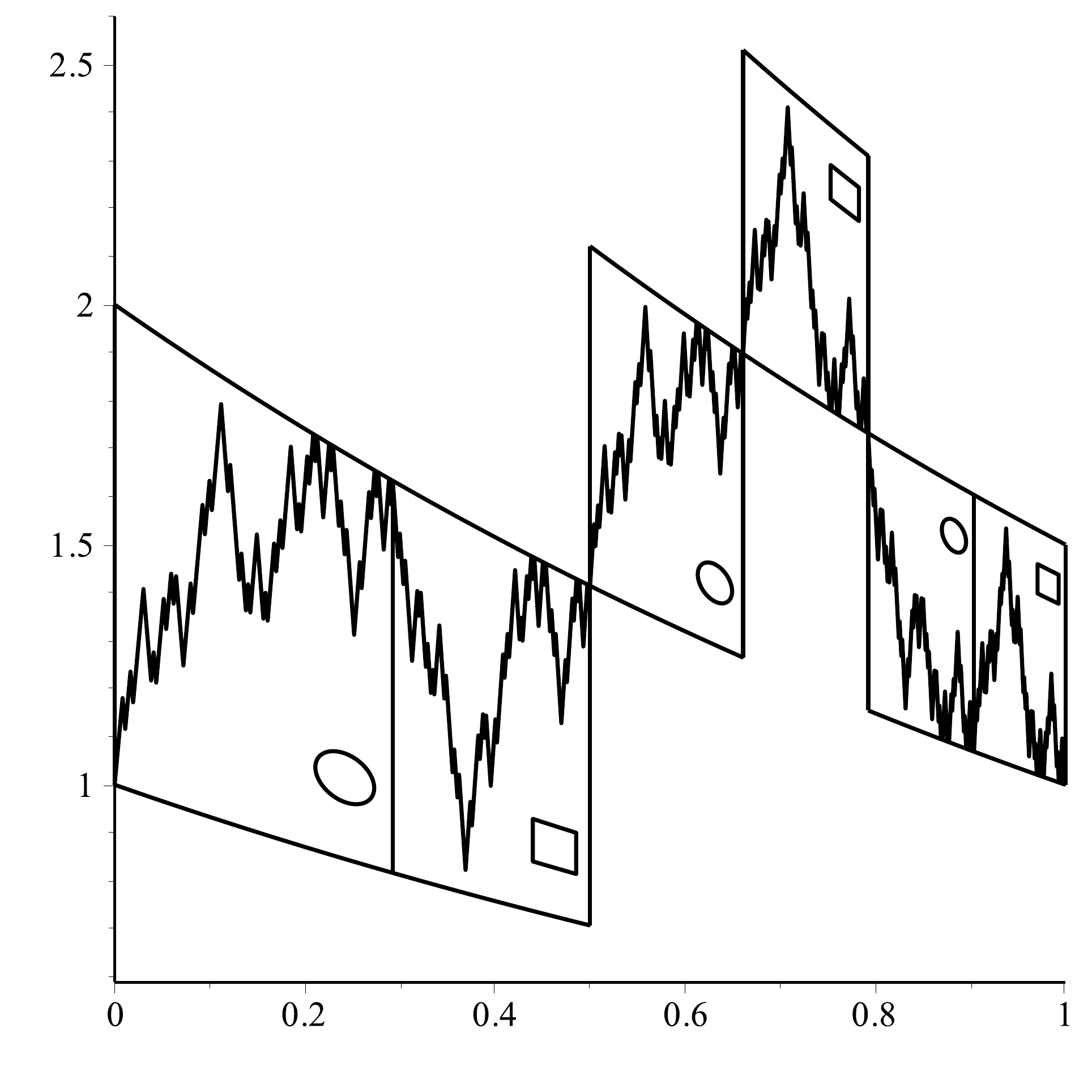}
\end{center} 
\caption{\label{dumas-rrseq-v4:fig:RudinShapiro_symmetry} The solution of
  the dilation equation (left-hand side) and the periodic function
  (right-hand side) associated to the Rudin--Shapiro sequence of
  Ex.~\ref{dumas-rrseq-v4:ex:rudinshapiro} show some symmetries. They
  are  obvious for the first one but quite buried for the
  second. }
\end{figure} 

\begin{example}[Rudin--Shapiro sequence]\label{dumas-rrseq-v4:ex:rudinshapiro}
  The Rudin--Shapiro sequence may be defined as $u_n =(-1)^{e_{2\,;
  11}(n)}$ where $e_{2\,; 11}(n)$ is the number of (possibly
  overlapping) occurrences of the pattern~$11$ in the binary expansion
  of the integer~$n$~\cite{BrCa70}. It is $2$-rational: it admits the
  generating family ~$(u_n,u_{2n+1})$ and the reduced linear
  representation, insensitive to the leftmost zeroes,
\[
L=\left(\begin{array}{cc}
1 & 1
	\end{array}\right),\qquad
A_0=\left(\begin{array}{cc}
1 & 1 \\ 0 & 0
	\end{array}\right),\qquad
A_1=\left(\begin{array}{cc}
0 & 0 \\ 1 & -1
	\end{array}\right),\qquad
C=\left(\begin{array}{c}
1 \\ 0
	\end{array}\right).
\]

As an application of   Theorem~\ref{dumas-rrseq-v4:thm:Theorem}, we obtain the 
  asymptotic expansion for the Rudin--Shapiro sequence~\cite{BrErMo83}
\[
\sum_{n\leq N}u_n\mathop{=}_{N\to+\infty}\sqrt{N}\Phi(\log_4N)+O(1)
\]
where~$\Phi$ is the $1$-periodic function defined by
$\Phi(t)=2^{1-\{t\}}F(4^{\{t\}-1})$ and~$F$ is defined through a
dilation equation for radix~$4$.  More precisely, we use the radix~$4$ linear representation $L' = L$, $A'_0 = A_0^2$, $A'_1 = A_0A_1$, $A_2 = A_1A_0$, $A'_3 = A_1^2$, $C' = C$. 
Functions~$F$ and~$\Phi$ are
illustrated in Figure~\ref{dumas-rrseq-v4:fig:RudinShapiro_symmetry}. Let us
denote $\pi_k$ the part of the graph of~$F$ which corresponds to the
interval $[k/8,(k+1)/8]$ for $0\leq k<8$. It is evident that the parts
of odd index on one hand and the part of even index on the other hand
reproduce the same pattern (with a piece upside down). This can be proved elementarily, playing at the same time with the radix~$2$ representation and the radix~$4$ representation. 

The symmetries of~$F$ are translated to~$\Phi$, but they
lose their graphical evidence. The pieces~$\pi_0$ and~$\pi_1$
disappear and the pieces~$\pi_k$, $2\leq k<8$ become the
pieces~$\pi_k'$ associated to the intervals
$[\log_4(k/2),\log_4((k+1)/2)]$. The links between the pieces become
more intricate. For example the pieces $\pi_2'$ and $\pi_3'$ on one
side and $\pi_6'$ and $\pi_7'$ on the other side are linked by the
formula $2^s\Phi(s)+2^t\Phi(t)=4$ under the condition that both
numbers $s\in[0,1/2]$ and $t\in[\log_43,1]$ are related by
$4^t-4^s=2$.
\end{example}

\section{Background}\label{dumas-rrseq-v4:sec:LinearvsAnalytic}

\subsection{Context}

Sequences related to a numeration system have a long history, but for what we are concerned the first studies about their asymptotic behaviour appear in the middle of the $20$th century. The studies in question deal with one example at a time and use elementary methods.
In our opinion, the most noteworthy article is the Brillhart,  Erd{\H{o}}s, and Morton's work~\cite{BrErMo83}, for it contains the seeds of almost all ideas about the topic: dilation equation, asymptotic expansion with periodic coefficients, H{\"o}lder continuous and non-differentiable functions, bounds for the periodic functions, Fourier series, Fourier coefficients, Dirichlet series, residues, Mellin transform. It can be and needs to be read and re-read many times to extract all its very substance. 

A more systematic study has begun in the late seventies, with two independent methods. One is related to the combinatorics of words~\cite{CoVa86,DuTh89,Dumont90,DuSiTh99}, and is summarized in~\cite[Sec.~4.1 and Sec.~6.1]{BaBeLiTh06}. The other is  based on analytic number theory~\cite{FlRa80,FlGo94,FlGrKiPrTi94,GrHw05}. Its most evolved version is~\cite{DrSz11}, even it is essentially limited to positive sequences, for it mixes several radices.  

Until recently,  the idea of radix-rational sequence was not emphasized, but this property sometimes rises to the surface~\cite[{\rm Corollaire}, p.~13-07]{BeFa78}, \cite[{\rm Satz}~3 and {\rm Satz}~4]{BrMo78}. This is not surprising for it was defined only twenty years ago~\cite{AlSh92}. The sequences are viewed only as satisfying divide-and-conquer recurrences without this wording being necessarily employed. Nevertheless, a linear representation appear in~\cite{GrHePr05}, where a new approach based on Fourier transforms and probability theory is used (this last point seems to limit the method to nonnegative linear representations). Besides,  dilation equations arise in~\cite{BeKr00,AlKa2006,Kruppel2009} and are systematically used in~\cite{Girgensohn2011}. But these authors limit themselves to scalar-valued functions. In comparison, \cite{GoKeSp92} uses matrices, eigenvalues and dilation equations, but the study is very specific to the Thue-Morse sequence.

Stolarsky~\cite{Stolarsky77} gives an extended bibliography (a little old-looking due to the date of publication, but rich) about digital sequences, that is sequences based on a numeration system.  More recent bibliographies may be found in~\cite{DrGa98} or \cite{DrGr10}. Moreover, Stolarsky notes: \emph{Whatever its mathematical virtues, the literature on sums of digital sums reflects a lack of communication between researchers}, a sentence which is yet topical. This was a motivation for us to provide to the reader an extended bibliography with the hope that this domain of research can get organized  and that the basic results about that topic will be known of all interested people.

\subsection{Linear algebra versus Analytic number theory}

Our point of interest here is the comparison between the analytic number theory method and our linear algebra method. The analytic method is based on meromorphic functions and computation of some residues. A good and concise account of the method is given  in~\cite[Sec.~8.2.3]{DrGr10}. The underlying idea is geometrically obvious but the application remains tricky. By contrast our approach is elementary but less intuitive. Nevertheless, it would be a mistake to oppose both approaches. Each has its own merit and the algebraic approach can greatly help the analytic approach.  We will emphasize this point with our favorite example.

\begin{example}[Dichopile algorithm, another way]\label{dumas-rrseq-v4:ex:Dichopile-7}
The use of the analytic number theory approach begins with the Dirichlet series
\begin{equation*}
	u(s) = \sum_{n\geq 1} \frac{\nabla f_n}{n^s}
\end{equation*}
associated with the sequence $u_n = \nabla f_n$, or better with the family of Dirichlet series associated with each of the sequences in the basis used to define the linear representation (Ex.~\ref{dumas-rrseq-v4:ex:Dichopile-1}). All these Dirichlet series are gathered into a row-vector valued Dirichlet series~$U(s)$, associated to the row-vector valued sequence~$U_n$ whose components are the sequences of the basis. Consequently the component of interest is recovered by $u(s) = U(s) C$.

The first question is to evaluate the abscissa of absolute convergence of~$U(s)$. Usually this results from an \emph{ad hoc} computation, usually a bound obtained by elementary arguments in a former work~\cite{BrErMo83}, \cite[Formul{\ae}~(6.7), (6.10)]{FlGrKiPrTi94}. For us this results from the computation of the joint spectral radius. The value $\rhojsr = 1$ shows that all  sequences in the basis are~$O(n^{\varepsilon})$ for every $\varepsilon > 0$ and this asymptotic relationship is the best possible (among the comparisons with a power of~$n$), so that the abscissa of absolute convergence is $\sigma_a = 1$. 

Next, we filter the integer~$n$ in the sum that defines~$U(s)$ according to their parity. This provides us 
with 
\[
	U(s) = \sum_{n=1}^{+ \infty}\frac{U_{2n}}{(2n)^s} + U_1 + \sum_{n=1}^{+ \infty} \frac{U_{2n+1}}{(2n+1)^s} 
	=
	\frac{1}{2^s}  U(s) A_0 + U_1 + \frac{1}{2^s}  U(s) A_1 + \sum_{n=1}^{+ \infty} \left(\frac{1}{(2n+1)^s} -\frac{1}{(2n)^s}\right) U_n  A_1,
\]
or in other words
\begin{equation}\label{dumas-rrseq-v4:eq:MeromorphicFumctionEquation}
U(s)	(\Id_6 - 2^{-s}Q)  = \nabla U (s)\qquad \text{with} \quad \nabla U (s) = U_1 + \sum_{n=1}^{+ \infty} \left(\frac{1}{(2n+1)^s}-\frac{1}{(2n)^s}\right) U_n A_1 .
\end{equation}
The last series converges absolutely for complex numbers~$s$ with a positive real part, because of the difference $(2n+1)^{-s} - (2n)^{-s}$, which is of order $n^{-s-1}$. 
 Equation~\eqref{dumas-rrseq-v4:eq:MeromorphicFumctionEquation} shows that~$U(s)$ extends on the half-plane $ \sigma > 0$
as a meromorphic function. More precisely the poles are the logarithms to base~$2$ of the eigenvalues of the matrix~$Q$. As a consequence, on the boundary of the half-plane of absolute convergence there is a vertical line of poles  $1 + \chi_k$ with $\chi_k = 2k\pi i/\ln 2$ and~$k$ integer, associated to the dominant eigenvalue~$2$. There are no other poles in the vertical strip between~$0$ and~$1$. 

\medskip

The next stage is the use of the Mellin-Perron summation formula~\cite[Th.~13]{HaRi15} 
\begin{equation} \label{dumas-rrseq-v4:eq:Perron1}
	\sum_{1 \leq k < N }\hspace*{-0.3em} U_k +\frac{1}{2}U_N =
	\frac{1}{2 \pi i} \int_{(c)} U(s) \, N ^s
\frac{ds}{s}.
\end{equation}
In this formula the integral is taken along a vertical line~$(c)$ at an abscissa~$c$ larger than the abscissa of absolute convergence. We express the function~$U(s)$ using Formula~\eqref{dumas-rrseq-v4:eq:MeromorphicFumctionEquation} and to emphasize the poles at abscissa 1, we write $ (\Id_6 - 2^{-s}Q)^{-1} = R(s) /  \left(1 -  {2}^{1-s} \right) ^{2} $ with~$R(s)$ a meromorphic function,  analytic on the right of~$0$. 
The function to be integrated has an expansion near $s = 1 + \chi_k$ of the form
\begin{equation*}
	\frac{R(s) \nabla U(s)}{\left(1 -  {2}^{1-s} \right) ^{2}}  \frac{ N ^s}{s}
	\mathop{=}_{s \to 1 + \chi_k}
	c_{2,k} N^{1+\chi_k}\frac{1}{(s - 1 - \chi_k)^2} +
	\left(c_{1,k} N^{1+\chi_k} \ln(N) + c_{0,k} N^{1+\chi_k}  \right) \frac{1}{s - 1 - \chi_k}  + O(1).
\end{equation*}
With Cauchy's residue theorem we change Formula~\eqref{dumas-rrseq-v4:eq:Perron1} into
\begin{equation}\label{dumas-rrseq-v4:eq:Pseudo-Perron}
	\sum_{1 \leq k < N }\hspace*{-0.3em} U_k +\frac{1}{2}U_N =
	\sum_{k = -\infty}^{+\infty} c_{1,k} N^{\chi_k} \times N \ln(N) + \sum_{k = -\infty}^{+\infty} c_{0,k} N^{\chi_k} \times N +
	\frac{1}{2 \pi i} \int_{(\varepsilon)} U(s) \, N ^s \frac{ds}{s},
\end{equation}
where this time~$\varepsilon$ is between~$0$ and~$1$. The term $U_N/2$ and the integral are of order $N^{\varepsilon}$. Overall~$N^{\chi_k}$ is nothing but $\exp(2 k \pi i \log_2 N)$ and some trigonometric series appear,
\begin{equation} \label{dumas-rrseq-v4:eq:Pseudo-expansion}
	\sum_{1 \leq k \leq N }\hspace*{-0.3em} U_k = 
	N \ln(N) \sum_{k = -\infty}^{+\infty} c_{1,k} \exp(2 k \pi i \log_2 N)
	+ N \sum_{k = -\infty}^{+\infty} c_{0,k} \exp(2 k \pi i \log_2 N) + O(N^{\varepsilon}).
\end{equation}
We have obtained an asymptotic expansion in the scale $N^\alpha \ln^\beta(N)$ with variable coefficients. The method is simple, concrete, obvious, natural: we see the poles, we see the line of integration, we push the line to the left, we catch the residues, and we have the asymptotic expansion of the partial sum.

\medskip

Unfortunately the radiant sun that illumines this method was soon overshadowed by thick clouds. First, we do not know if the points $1 + \chi_k$ are really poles of~$U(s)$. In this example, the only point that is certainly a pole is the abscissa of absolute convergence $\sigma_a = 1$, according to Landau's theorem~\cite[Th.~10]{HaRi15} , because the matrices~$L$, $A_0$ and~$A_1$ have nonnegative coefficients so that all the components of~$U_n$ are nonnegative. The same phenomenon occurs with the constant sequence of value~$1$, whose associated Dirichlet series is the Riemann zeta function~$\zeta(s)$. In that case, Equation~\eqref{dumas-rrseq-v4:eq:MeromorphicFumctionEquation} is the usual link between~$\zeta(s)$ and the alternate Riemann zeta function. Because of the factor $1- 2^{1-s}$ it seems that~$\zeta(s)$ has a line of poles $1 + \chi_k$, but we know that the only pole is~$1$. Here, comparing~\eqref{dumas-rrseq-v4:eq:Pseudo-expansion} and our expansion~\eqref{dumas-rrseq-v4:eq:DichopileExpansion}, which begins with a dominant term $N \log_2 (N) /2$, we see that~$1$ may be a double pole but that the other points $1 + \chi_k$ are certainly at most simple poles.  

Second, there is no reason (at this stage) for the trigonometric series above, which is the coefficient of~$N$ in~\eqref{dumas-rrseq-v4:eq:Pseudo-expansion}, to be a convergent series and defines a periodic function. Usually, an extra argument proves the occurrence of a periodic function (see for example~\cite[Lemma~H]{FlRa80}, \cite[Sec.~2]{Coquet83} or \cite[Sec.~2]{BrErMo83}).  For us, Formula~\eqref{dumas-rrseq-v4:eq:DichopilePeriodicFunction} defines a periodic function.

Third, the order of growth of~$U(s)$ along a vertical line does not permit the use of the Mellin-Perron formula.  Equation~\eqref{dumas-rrseq-v4:eq:MeromorphicFumctionEquation} shows, by the study of the right-hand side, that $\abs{U(\sigma+i t)}$ does not grow at infinity more rapidly than $\abs{t}^0$ for $\sigma > 1$  and than $\abs{t}^1$ for $0 < \sigma < 1$, hence no more rapidly than $\abs{t}^{1-\sigma}$ for $0 < \sigma < 1$. It is a consequence of the Lindel{\"o}f theorem~\cite[Sec.~III.4]{HaRi15}, which asserts that the order of growth is a convex function with respect to~$\sigma$.   The bound is not enough small to guarantee the absolute convergence of the integral on the line~$(\varepsilon)$. However the Mellin-Perron formula, say of the second order, 
\begin{equation*} 
\frac{1}{N }\sum_{1\leq k\leq n<N }\hspace*{-1em}U_k=
\sum_{1 \le k < N }\hspace*{-0.3em} U_k\left(1-\frac k N \right)
= \frac{1}{2 \pi i} \int_{(c)} U(s) \, N ^s
\frac{ds}{s(s+1)}
\end{equation*} 
can be used because if we change the line of integration~$(c)$ into~$(\varepsilon)$ the integrand becomes~$O(\abs{t}^{-1-\varepsilon})$. We obtain an expansion
\begin{equation*}
	\frac{1}{N }\sum_{1\leq k\leq n<N }\hspace*{-1em}u_k \mathop{=}_{N \to +\infty}
	N \ln(N) \Psi_1(\log_\base N) + N \Psi_0(\log_\base N) + O(N^\varepsilon),
\end{equation*} 
where~$\Psi_1(t)$ and~$\Psi_0(t)$ are $1$-periodic functions defined as sums of convergent trigonometric series. 
But Proposition~6.4 of~\cite{FlGrKiPrTi94} provides us, by summation of the expansion~\eqref{dumas-rrseq-v4:eq:DichopileExpansion}, with
\begin{equation}\label{dumas-rrseq-v4:eq:SummationDichopileExpansion}
	\frac{1}{N }\sum_{1\leq k\leq n<N }\hspace*{-1em}u_k \mathop{=}_{N \to +\infty}
	\frac{1}{4} N \ln(N) -\frac{1}{8} N +  N \Psi(\log_\base N) + o(N).
\end{equation} 
The uniqueness of the asymptotic expansion shows that $\Psi_1(t) = 1/4$ and $\Psi_0(t) = -1/8 + \Psi(t)$.  Moreover Proposition~6.4 of~\cite{FlGrKiPrTi94} gives the link between the Fourier coefficients~$c_k(\Phi)$ of  the function~$\Phi(t)$, which appears in the asymptotic expansion~\eqref{dumas-rrseq-v4:eq:DichopileExpansion}, and the Fourier coefficients~$c_k(\Psi)$ of the  function~$\Psi(t)$ in~\eqref{dumas-rrseq-v4:eq:SummationDichopileExpansion},
\begin{equation*}
	c_k(\Phi) = (2+\chi_k) c_k(\Psi).
\end{equation*}
This enables us  to show that the Fourier series of~$\Phi(t)$ is the trigonometric series which appears in~\eqref{dumas-rrseq-v4:eq:Pseudo-expansion} as the coefficient of~$N$. In other words, to push the line of integration on the left and collect the residues provides us  with the right Fourier series, even if this first seems to be a wrong process.   Bernstein theorem~\cite[Vol.~I, p.~240]{Zygmund02}, quoted in~\cite[Prop.~6]{GrHw05},  guarantees the uniform convergence of the Fourier series for a H{\"o}lder continuous function with exponent $> 1/2$. We know by our algebraic approach that it is the case for the dichopile algorithm. In this example, the Fourier series converges, but it is not the general case, and it can be necessary to consider F{\'e}jer sums. 
\end{example}

To summarize, we have two methods. One is brilliant but needs dexterity (\cite{Hwang98} provides a good example). The other is more pedestrian and more accessible. Moreover, the algebraic approach provides arguments to sustain the analytic approach. 

\section{Fourier coefficients}\label{dumas-rrseq-v4:sec:FourierCoefficients}

Having made the link between the $1$-periodic function and its Fourier series, we still have to compute its Fourier coefficients. At this point there are two possibilities. The first one is the fine case, where the Dirichlet series~$U(s)$ (notation of Example~\ref{dumas-rrseq-v4:ex:Dichopile-7})
is explicitly known, as in many examples dealt with in~\cite{FlGrKiPrTi94} or~\cite{GrHw05}, which use the Riemann zeta function. The second possibility is the generic one and the dichopile algorithm enters in this case. This is the case on which we will lay stress.

\subsection{Residues}

We have at our disposal several methods to compute numerically the Fourier coefficients. The first method is the direct application of the analytic number theory approach. The Fourier coefficients are obtained through some residues. 

\begin{example}[Dichopile algorithm, computation of the Fourier coefficients-1]
The basic formula is~\eqref{dumas-rrseq-v4:eq:MeromorphicFumctionEquation}, that is $U(s)	(\Id_6 - 2^{-s}Q)  = \nabla U (s)$. The idea is to  compute the right member for a pole~$1 + \chi_k$ of~$U(s)$ to obtain the residue at~$1 + \chi_k$. It is practically easier to use the change of basis of Example~\ref{dumas-rrseq-v4:ex:Dichopile-3}, that is to use a basis for which the matrix~$Q$ is in Jordan form. This emphasizes the components of~$U(s)$ which are really involved with the line of poles~$1 + \chi_k$. 

However, the computation of $\nabla U(s)$ is not so easy, because the convergence of the series in~\eqref{dumas-rrseq-v4:eq:MeromorphicFumctionEquation} is slow. Grabner and Hwang~\cite{GrHw05} deal with examples where they speed up the convergence by considering differences of the second order in place of the first order. Their goal is to ease the process of pushing the line described in Example~\ref{dumas-rrseq-v4:ex:Dichopile-7}. Overall, they use their so-called \emph{$1/2$-balancing principle}. We will not insist on this point because it is well explained in~\cite{GrHw05}. Essentially it leads to consider series with a factor~$1/4^m$ in place of series with a factor~$1/2^m$, hence a clear gain for the convergence speed. It must be noticed that the description of the method with a factor~$1/16^m$ is rather optimistic. It works only for very specific examples, and a factor~$1/4^m$ is the general case. 
\end{example}

\subsection{Crude method}

The second method is to return to the definition of the Fourier coefficients. We begin  without subtlety.

\begin{example}[Dichopile algorithm, computation of the Fourier coefficients-2]
By definition, the Fourier coefficients of~$\Phi(t)$ are
\begin{equation*}
	c_k = \int_0^1 \Phi(t) e^{-2\pi i k t}\,dt.
\end{equation*}
The change of variable $t = 1 + \log_2 x$ with $1/2 \leq x \leq 1$ transforms this formula into 
\begin{equation*}
c _k = \frac{1}{\ln 2} \int_{1/2}^1 \Phi(\log_2 x) e^{- 2 k \pi i \log_2 x} \, \frac{dx}{x} =
	c_k = \frac{1}{\ln 2} \int_{1/2}^1 \left( 1 - \frac{1}{2}\log_2 x + \frac{g_5(x)}{x} \right) e^{- 2 k \pi i \log_2 x} \, \frac{dx}{x}.
\end{equation*}
Distinguishing the case $k = 0$, we have
\begin{equation}\label{dumas-rrseq-v4:eq:FourierCoefficientsAsMellinEvaluations}
c_0 = \frac{5}{4} + \frac{1}{\ln 2} \int_{1/2}^1 \frac{g_5(x)}{x^2} \, dx,
\qquad
c_k = -\frac{i}{4k \pi} + \frac{1}{\ln 2} \int_{1/2}^1 \frac{g_5(x)}{x^2} e^{- 2 k \pi i \log_2 x} \, dx
\end{equation}
A crude approach is to compute these integrals by the trapezoidal rule, with the nodes~$j/2^K$ for $2^{K-1} \leq j \leq 2^K$ and a given~$K$. Obviously, this uses the cascade algorithm. As~$g_5(x)$ is H{\"o}lder with exponent~$\alpha$ for every $0 < \alpha < 1$, the error is of order~$1/2^{K\alpha}$ for every $0 < \alpha <1$, that is essentially~$1/2^K$. This method has the advantage of the simplicity, but we cannot speed up the computation {\`a} la Richardson because there is no asymptotic expansion of the error. It gives only a rough estimation. Below are the values obtained with $K = 10$ on the left-hand side and with $K = 12$ on the right-hand side. The computations have been made with $15$ digits to avoid the rounding errors. The correct digits are written in bold.

\begingroup\scriptsize
\begin{equation*}\setlength\arraycolsep{1pt} 
	\begin{array}{lcl}
c_0 & \simeq & -0.\bm{362}6476334 \\ 
c_{1} & \simeq & +0.\bm{00328}58226  +0.\bm{0019}776043\,i \\ 
c_{2} & \simeq & +0.\bm{003069}5432  -0.\bm{00062}87349\,i \\ 
c_{3} & \simeq & +0.\bm{001685}7421  +0.\bm{0012}276124\,i \\ 
c_{4} & \simeq & +0.\bm{000542}4079  -0.\bm{0011}900529\,i \\ 
c_{5} & \simeq & +0.\bm{001132}8340  +0.\bm{0005}285686\,i \\ 
c_{6} & \simeq & +0.\bm{000616}9547  +0.\bm{0003}775004\,i \\ 
c_{7} & \simeq & +0.\bm{000678}9418  -0.\bm{0004}749207\,i \\ 
c_{8} & \simeq & -0.\bm{00031}21218  +0.\bm{0002}826099\,i \\ 
c_{9} & \simeq & +0.\bm{000338}7154  -0.\bm{000}4014520\,i \\ 
c_{10} & \simeq & +0.\bm{000502}6412  +0.\bm{000}0947029\,i \\ 
\end{array}\qquad\qquad\qquad\qquad
%
	\begin{array}{lcl}
c_0 & \simeq & -0.\bm{3627}354935 \\ 
c_{1} & \simeq & +0.\bm{003284}7914  +0.\bm{0019}880642\,i \\ 
c_{2} & \simeq & +0.\bm{0030691}945  -0.\bm{00062}20152\,i \\ 
c_{3} & \simeq & +0.\bm{0016854}542  +0.\bm{00123}35442\,i \\ 
c_{4} & \simeq & +0.\bm{0005422}948  -0.\bm{00118}36225\,i \\ 
c_{5} & \simeq & +0.\bm{001132}5212  +0.\bm{00053}50851\,i \\ 
c_{6} & \simeq & +0.\bm{000616}6015  +0.\bm{00038}47058\,i \\ 
c_{7} & \simeq & +0.\bm{0006785}365  -0.\bm{00046}64313\,i \\ 
c_{8} & \simeq & -0.\bm{0003118}626  +0.\bm{00029}12762\,i \\ 
c_{9} & \simeq & +0.\bm{000338}4147  -0.\bm{00039}14805\,i \\ 
c_{10} & \simeq & +0.\bm{0005021}373  +0.\bm{00010}51334\,i \\ 
\end{array}
\end{equation*}
\endgroup

\noindent
This permits us to draw a picture (Fig.~\ref{dumas-rrseq-v4:fig:Dichopile-Phi-Function+FourierTruncated10}) of the truncated Fourier series against the periodic function. Because of the H{\"o}lderian character of~$\Phi(t)$ the difference between the function and the partial Fourier sum of order~$n$ is of order~$O(\ln(n)/n^\alpha)$ for every~$0 < \alpha < 1$ and the convergence is rather slow~\cite[Th.~II.10.8]{Zygmund02}. 
\end{example}

\begin{figure}[t]
	\centering
		\begin{center}
	\includegraphics[width=0.45\linewidth]{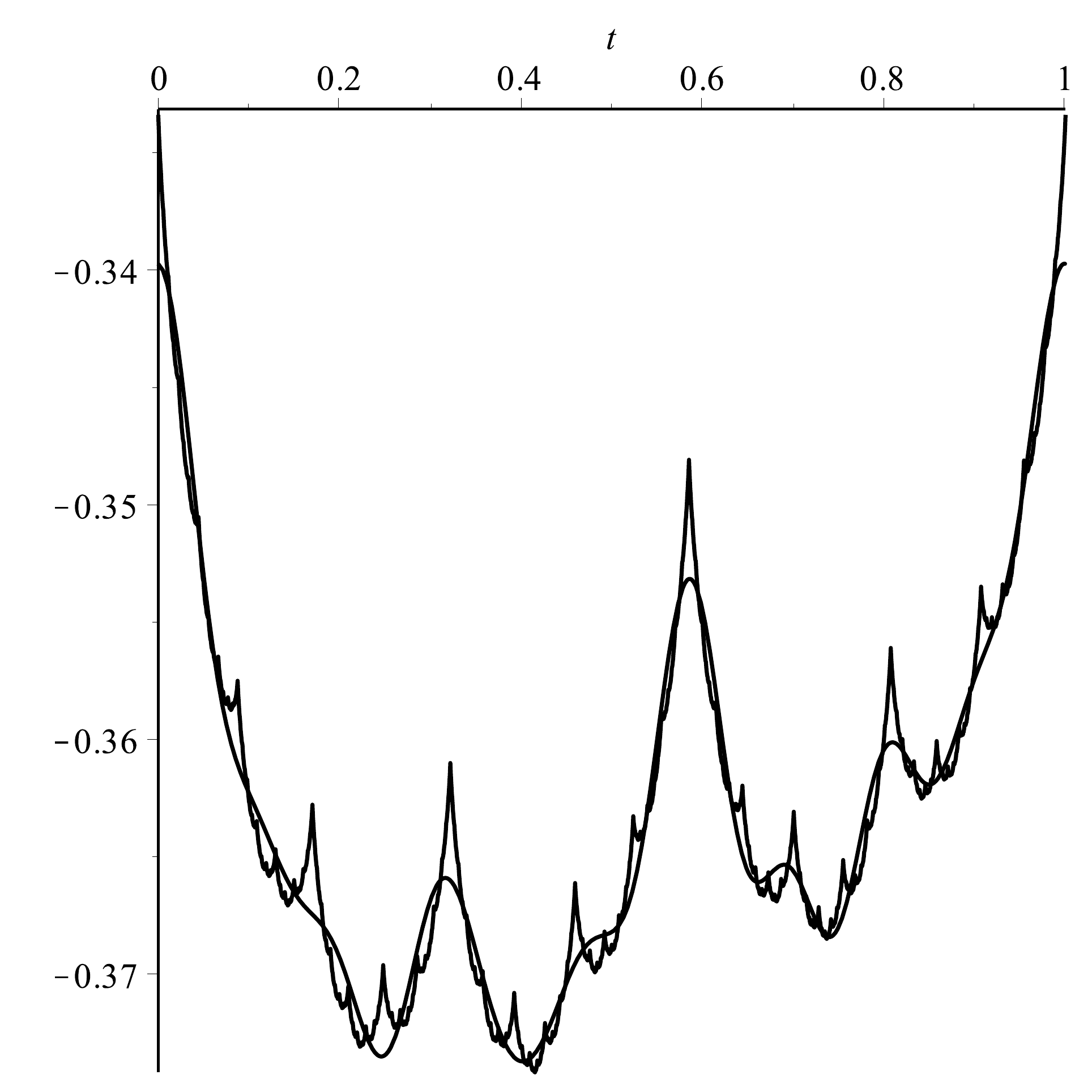}
\end{center}
	\caption{\label{dumas-rrseq-v4:fig:Dichopile-Phi-Function+FourierTruncated10}The periodic function~$\Phi(t)$ of the dichopile algorithm and its Fourier sum of order~$10$.}
\end{figure}

\subsection{Moments}

To go further, let us proceed to a detour. As the result we have in mind has some interest in itself, we consider a rather general framework. For a function~$F(x)$ continuous on the interval~$[0,1]$, we define its moments and partial moments, respectively
\begin{equation}\label{dumas-rrseq-v4:eq:MomentsDefiniton}
	M_\ell = \int_0^1 F(z) z^{\ell-1}\, dz, \qquad 
	M_{\ell,r} = \int_{r/\base}^{(r+1)/\base} F(z) z^{\ell-1}\, dz,
\end{equation}
where~$\ell$ is a positive integer and $0 \leq r < \base$ with~$\base \geq 2$. Remarkably for the solution of a dilation equation of the type considered in Section~\ref{dumas-rrseq-v4:sec:DilationEquations}, the moments and the partial moments can be computed exactly. Similar computations appear in~\cite{Evans57,LaTa92} (Cantor distribution)
and~\cite[Eq.~3.1 or~4.2]{ShYa93} (coefficients for wavelets)
or \cite[p.~5]{TaGiDo1998} (computation of a Riemann-Stieltjes integral).

\begin{lemma}\label{dumas-rrseq-v4:lemma:RecursionForMoments}
With the notations of Section~\ref{dumas-rrseq-v4:sec:DilationEquations} and under the assumption that none of the numbers $\base^\ell \rho$, $\ell \geq 1$, is an eigenvalue of~$Q$, the recursions ($W_b = A_b V$ for $0 \leq b < \base$)
\begin{equation}\label{dumas-rrseq-v4:eq:DichopileGeneralRecursionForMoments}
	\base^\ell M_\ell J - Q M_\ell = 
	\frac{1}{\ell} \sum_{b = 0}^{\base - 2} \left(\base^\ell - (b+1)^\ell\right)  W_b
	+ \sum_{k=1}^{\ell-1} \binom{\ell-1}{k-1} \left(\sum_{r=0}^{\base -1} r^{\ell-k}A_b\right) M_k,
\end{equation}
\begin{equation}\label{dumas-rrseq-v4:eq:DichopileGeneralRecursionForPartialMoments}
		M_{\ell,r} J = \sum_{b < r} \frac{(r+1)^\ell - r^\ell}{\ell \base^\ell} W_b 
	+ \frac{1}{\base^\ell} \sum_{k=1}^{\ell} \binom{\ell-1}{k-1} r^{\ell-k} M_k,\qquad 0 \leq r < \base,
\end{equation}
determine all the moments and partial moments of the matrix-valued function~$F(x)$ defined by Lemma~\ref{dumas-rrseq-v4:lemma:TheoremAboutDilationEquations}. 
\end{lemma}

\begin{proof}
Using the dilation equation~\eqref{dumas-rrseq-v4:eq:DilationEquation} we find readily~\eqref{dumas-rrseq-v4:eq:DichopileGeneralRecursionForPartialMoments} by a mere change of variable and the binomial theorem. 

It may look troublesome to have to solve Equation~\eqref{dumas-rrseq-v4:eq:DichopileGeneralRecursionForPartialMoments}, as it involves multiples of~$M_\ell$ on both the left hand and the right hand sides. However, we can emphasize the structure of the Jordan cell by writing it $J = \rho \matI_{\left[0,\nu\right)} + N$ and the structure of the~$d\times \nu$ matrix~$M_\ell$ by viewing it as the collection of its columns~$M_\ell^0$, $M_\ell^1$, $\ldots$, $M_\ell^{\nu-1}$. We add $M_\ell^{-1} = 0$ for convenience. With these notations~\eqref{dumas-rrseq-v4:eq:DichopileGeneralRecursionForPartialMoments} rewrites
\begin{equation}\label{dumas-rrseq-v4:eq:DichopileGeneralAndPracticalRecursionForMoments}
	(\base^\ell \rho \matI_d - Q) M_\ell^j =  - \base^\ell M_\ell^{j-1} + 
	\frac{1}{\ell} \sum_{b = 0}^{\base - 2} \left(\base^\ell - (b+1)^\ell\right)  W_b^j
	+ \sum_{k=1}^{\ell-1} \binom{\ell-1}{k-1} \left(\sum_{r=0}^{\base -1} r^{\ell-k}A_b\right) M_k^j, \qquad 0 \leq j < \nu, 
\end{equation}
because of the equality $M_\ell N = (0,M_\ell^0,M_\ell^1,\ldots, M_\ell^{\nu-2})$. It is now clear that  Equation~\eqref{dumas-rrseq-v4:eq:DichopileGeneralAndPracticalRecursionForMoments} enables us to compute successively $M_1^0$, $M_1^1$, $\ldots$,  $M_1^{\nu-1}$, $M_2^0$, $\ldots$, $M_2^{\nu-1}$ and more generally all the moments, if the numbers $\base^\ell \rho$ are not eigenvalues of~$Q$. Last, Equation~\eqref{dumas-rrseq-v4:eq:DichopileGeneralRecursionForPartialMoments} provides us with the partial moments because~$J$ is invertible. 
\end{proof}

\begin{example}[Dichopile algorithm, computation of the moments]
For the dichopile algorithm, Lemma~\ref{dumas-rrseq-v4:lemma:RecursionForMoments} translates into the formul{\ae}
\begin{equation}\label{dumas-rrseq-v4:eq:M0Recursion}
	(2^{\ell+1} \matI_6 - Q) M^0_\ell = \frac{2^{\ell} - 1}{\ell} A_0V_2^0 +  A_1 \sum_{j=1}^{\ell-1} \binom{\ell-1}{j-1} M_j^0,
\end{equation}
\begin{equation}\label{dumas-rrseq-v4:eq:M1Recursion}
	(2^{\ell+1} \matI_6 - Q) M^1_\ell = \frac{2^{\ell} - 1}{\ell} A_0V_2^1 +  A_1 \sum_{j=1}^{\ell-1} \binom{\ell-1}{j-1} M_j^1 - 2^{\ell}M^0_\ell,
\end{equation}
\begin{equation*}
	2 M^0_{\ell,1} = \frac{1 - 1/2^{\ell}}{\ell } A_0 V^0_2 + \frac{1}{2^{\ell }} A_1 \sum_{j=1}^{\ell} \binom{\ell-1}{j-1} M^0_j,
\end{equation*}
\begin{equation}\label{dumas-rrseq-v4:eq:PartialMoment11}
	2 M^1_{\ell,0}  
 = \frac{1 -1 /2^{\ell}}{\ell  } A_0 V_2^1 + \frac{1}{2^{\ell}} A_1 \sum_{j=1}^\ell \binom{\ell-1}{j-1} M^1_j - M^0_{\ell,1},
\end{equation}
Since some components of~$f(x)$ and~$g(x)$ are known, we can verify the result of the computation.
\end{example}

\subsection{Mellin transform}

Let us return to the computation of the Fourier coefficients. If we look scrupulously to the proof of Theorem~\eqref{dumas-rrseq-v4:thm:Theorem} and particularly to Formula~\eqref{dumas-rrseq-v4:eq:ChuVandermonde}, we see that the occurrence of periodic functions in the expansion of the sequence comes from the term 
\begin{equation}\label{dumas-rrseq-v4:eq:PhiDefinition}
	\Phi(t) = \rho ^{1-\sawtooth{t}} e^{i \vartheta (1-\sawtooth{t})} F(\base^{\sawtooth{t}-1}) \binom{1-\sawtooth{t}}{p}
\end{equation}
for some nonnegative integer~$p$.  This integer~$p$ is $\ell -m$ in the Chu-Vandermonde type Formula~\eqref{dumas-rrseq-v4:eq:ChuVandermonde}. 

Let us focus on  the case $p = 0$. 
With the change of variable $t = 1 + \log_\base x$, $1/\base \leq x \leq 1$ (that already appears in Delange's article~\cite[p.~37]{Delange75}), the Fourier coefficients of the $1$-periodic function~$\Phi(t)$ are
\begin{equation*}
	C_k =  \int_0^1 \Phi(t) e^{- 2 \pi i k t}\, dt = 
	\int_{1/\base}^1 \rho^{\log_\base x} e^{i \vartheta \log_\base x} F(x) 
	e^{- 2 \pi i k \log_\base x} \frac{dx}{x \ln \base},
\end{equation*}
that is
\begin{equation*}
	C_k = \frac{1}{\ln \base} \int_{1/\base}^1 x^{-1 + \frac{\ln \rho + i \vartheta}{\ln \base}} F(x) 
	 x^{-\chi_k} \, dx,
\end{equation*}
with $\chi_k = 2k\pi i/\ln \base $. 
It turns out that these coefficients can be viewed as special values of a Mellin transform, namely values of
\begin{equation*}
	F^*(s) = \int_{1/\base}^1 x^{s-1} F(x) \,dx
\end{equation*}
at points $s = - \chi_k + \log_\base \lambda $ with $\lambda = \rho e^{i \vartheta}$ the eigenvalue of~$Q$ under consideration. 
According to Lemma~\ref{dumas-rrseq-v4:lemma:RecursionForMoments},  we conversely know the value of the Mellin transform at positive integers,
\begin{equation*}
	F^*(\ell) = \sum_{r=1}^{\base -1} M_{\ell,r}
\end{equation*}
with the notations of~\eqref{dumas-rrseq-v4:eq:MomentsDefiniton}. 

\begin{lemma}
The Mellin transform of the matrix-valued function~$F(x)$ defined by Lemma~\ref{dumas-rrseq-v4:lemma:TheoremAboutDilationEquations},
\begin{equation*}
	F^*(s) = \int_{1/\base}^1 x^{s-1} F(x) \,dx,
\end{equation*}
is an entire function, which can be expressed as the sum of a Newton series
\begin{equation}\label{dumas-rrseq-v4:eq:NewtonSeries}
	F^*(s) = \sum_{n=0}^{+\infty} \binom{s-1}{n}  \Delta^n F^*(1),
\end{equation}
where~$\Delta$ is the forward difference operator, acting on the variable~$s$ of~$F^*(s)$. A partial sum gives the value~$F^*(s)$ with an error of order the first neglected term.
\end{lemma}

\begin{proof}
It must be noticed that this function~$F^*(s)$ is  analytic in the whole complex plane. Actually a Mellin transform is an integral from~$0$ to~$+\infty$, and the behaviour of the function at the ends of the interval determines the vertical strip in the complex plane where the Mellin transform is defined and analytic. Here the integrand vanishes in the neighborhood of~$0$ and of~$+\infty$, so the transform is an entire function. We write the term $x^{s-1}$ as $(1 - (1-x))^{s-1}$ and we expand it by the binomial theorem. This gives 
\begin{equation*}
	F^*(s) = \sum_{n=0}^{+\infty} \binom{s-1}{n} \int_{1/\base}^1 F(x) (x-1)^n \, dx .
\end{equation*}
Again with the binomial theorem, the last integral appears as the $n$th order difference of the sequence $F^*(\ell)$ evaluated at $\ell = 1$. In other words we obtain Equation~\eqref{dumas-rrseq-v4:eq:NewtonSeries}. Moreover the integral expression of~$\Delta^n F^*(1)$ permits to show that it behaves as $(F(1-1/\base) + o(1)) (1-1/\base)^{n+1}/(n+1)$. As a consequence, the series converges at least as fast as a geometric series of ratio $1-1/\base$. Hence the assertion about the error term.
\end{proof}

The above result may suggest that the calculation of the Fourier coefficients is particularly simple. This is not quite true, as the following example shows.

\begin{example}[Dichopile algorithm, computation of the Fourier coefficients-3]
	According to~\eqref{dumas-rrseq-v4:eq:FourierCoefficientsAsMellinEvaluations}, we have to compute the integrals
\begin{equation*}
	e_k = \int_{1/2}^1 g_5(x) x^{-\chi_k -2}\,dx = g_5^*(-1 -\chi_k )
\end{equation*}
and~$g_5^*(s)$ expresses as
\begin{equation*}
	g_5^*(s) = \sum_{n=0}^{+\infty} \binom{s-1}{n}  \Delta^n g_5^*(1) =  \sum_{n=0}^{+\infty} \binom{s-1}{n} \Delta^n M^1_{1,1},\quad 
\end{equation*}
where~$\Delta$  acts on the first lower index~$\ell$ of~$M^1_{\ell,1}$. At this point, two problems arise. 
First the moments~$M^1_{\ell,1}$ are of order~$1/\ell$, while the differences~$\Delta^n M^1_{1,1}$ are of order $2^n/n$ according to the proof of the above proposition. Hence there is a strong cancellation in the computation of differences~$\Delta^n M^1_{1,1}$. To take this phenomenon into account, we do not compute these differences as float numbers but exactly as rational numbers they are. Second for $s = -1 + \chi_k $ the modulus of $\binom{s-1}{n} /(2^nn)$ increases first up to a maximal value   of order 
\[
10^{2.06
\,|k|- 0.65
\ln |k|-1.08
}
\]
obtained for $n \simeq 5 |k|$, and next decreases towards~$0$. Henceforth,  to obtain the sum within an error $10^{-m}/2$ we have to sum the terms until they become smaller than~$10^{-m}/2$ and 
we use mantiss{\ae} of length $2.06\, |k| -0.65 \ln |k| -1 + m + 5$ (with~$5$ digits to control the rounding errors). Practically, we sum the series up to the $n$th term with $n \simeq (\pi/\ln(2))^2 \, |k| + m \log_2{10} - 3/2\, \log_2 |k|  \simeq 20.5\, |k| + 3.3 \, m - 2.2 \, \ln |k| $. This analysis is not valid for the case $k = 0$, where the absolute value of the sequence is decreasing from the beginning.  
Below are the first few values of the Fourier coefficients for $m=50$.

\begingroup\scriptsize
\begin{equation*}\setlength\arraycolsep{1pt}
	\begin{array}{lcl}
c_0 & \simeq & -0.36276483219909523733941579131627817954357682261599 \\ 
c_{1} & \simeq & +0.00328444028368975642383395704527876596759317794466  +0.00199132072044919779610043930356670834266095959747\,i \\ 
c_{2} & \simeq & +0.00306910005017327177913457733939403086029964937936  -0.00062027605748476162777330653703757401619891385239\,i \\ 
c_{3} & \simeq & +0.00168540331698581529572968437256156261586371542757  +0.00123480183822010180331253966429137775847524941664\,i \\ 
c_{4} & \simeq & +0.00054226937333349800444215803839172178044911364862  -0.00118254688336345399259102335961528247150067405496\,i \\ 
c_{5} & \simeq & +0.00113248887958587730781507505377755682399833345293  +0.00053603714268075411158841686643383757932219955252\,i \\ 
c_{6} & \simeq & +0.00061656997916724901585530584609901168042150761797  +0.00038561742152620206710596427624390201335494057698\,i \\ 
c_{7} & \simeq & +0.00067850368090861766393827715162354538698380890834  -0.00046549574088176188102491212188129989781189682338\,i \\ 
c_{8} & \simeq & -0.00031184978206237104092864408642405259547335766921  +0.00029217720780737790487915974566745439206958256568\,i \\ 
c_{9} & \simeq & +0.00033839108394137885396861099363986939271517660809  -0.00039052852080456028985072135518517264916825572241\,i \\ 
c_{10} & \simeq & +0.00050210083984953064910235351056678436057811735044  +0.00010608705990659151346295572715633259834904284986\,i \\ 
	\end{array}
\end{equation*}
\endgroup 

\end{example}

We have dealt with the case $p = 0$ in Equation~\eqref{dumas-rrseq-v4:eq:PhiDefinition}. The general case of a nonnegative integer~$p$ would lead us to consider derivatives of the Mellin transform~$F^*(s)$. 

\small

\bibliographystyle{model1-num-names}
\bibliography{\bibpath/RRsequence}

\end{document}